\documentclass{article}
\usepackage[utf8]{inputenc}
\usepackage{amsfonts, amssymb, amsmath, amsthm, euscript}
\usepackage{graphicx, subfigure, bm}
\usepackage[margin=0.9in]{geometry}
\usepackage{color}
\usepackage[colorlinks,citecolor=blue]{hyperref}
\usepackage{bbm}
\usepackage[english]{babel}
\usepackage{enumitem} 
\usepackage{soul}
\usepackage{mathabx}
\usepackage[toc,page]{appendix} 
\usepackage{mathtools}
\mathtoolsset{showonlyrefs}
\usepackage{calligra}
\DeclareMathAlphabet{\mathcalligra}{T1}{calligra}{m}{n}
\usepackage{algorithm, algpseudocode}

\usepackage{tikz}
\usetikzlibrary{arrows.meta, positioning}

\newcommand{\E}{\mathbb{E}}

\newcommand{\triplenorm}[1]{{\vert\kern-0.25ex\vert\kern-0.25ex\vert #1 
    \vert\kern-0.25ex\vert\kern-0.25ex\vert}}

\newtheorem{thm}{Theorem}[section]
\newtheorem{prop}{Proposition}[section]

\newtheorem{definition}{Definition}[section]
\newtheorem{asmp}{Assumption}
\newtheorem{example}{Example}[section]
\newtheorem{remark}{Remark}
\theoremstyle{definition}
\theoremstyle{remark}

\newcommand{\mathz}{\ooalign{$z$\cr\hfil\rule[.5ex]{.2em}{.06ex}\hfil\cr}}

\newcommand{\xdownarrow}[1]{%
  {\left\downarrow\vbox to #1{}\right.\kern-\nulldelimiterspace}
}
\date{}

\begin{document}

\title{Relative Arbitrage Opportunities with Interactions among $N$ investors \footnote{The authors are very thankful to the two anonymous referees and the Editors for their careful reviews and suggestions.}}
\author{Tomoyuki Ichiba\thanks{Department of Statistics and Applied Probability, South Hall, University of California, Santa Barbara, CA 93106, USA (E-mail: \href{mailto:ichiba@pstat.ucsb.edu}{ichiba@pstat.ucsb.edu}) Part of research is supported by National Science Foundation grants DMS-1615229 and DMS-2008427.} \and Nicole Tianjiao Yang\thanks{Department of Mathematics, University of Tennessee, Knoxville, TN 37996, USA (E-mail: \href{mailto:tyang23@utk.edu}{tyang23@utk.edu})}}
\maketitle
\begin{abstract}
The relative arbitrage portfolio outperforms a benchmark portfolio over a given time-horizon with probability one. With market price of risk processes depending on the market portfolio and investors, this paper analyzes the multi-agent optimization of relative arbitrage opportunities in the coupled system of market and wealth dynamics. We construct a well-posed market dynamical system of McKean-Vlasov type under an empirical measure of investors, where each investor seeks for relative arbitrage with respect to a benchmark dependent on market and all the agents. We show the conditions to guaranty relative arbitrage opportunities among competitive investors through the Fichera drift. 
Under mild conditions, we derive the optimal strategies for investors and the unique Nash equilibrium that depends on the smallest nonnegative solution of a Cauchy problem.
\end{abstract}

\section{Introduction}
\label{intro}
\numberwithin{equation}{section}
Market participants usually compare the performance of an investment strategy with a benchmark index. Among different metrics and tools for capturing opportunities that outperform a benchmark portfolio, relative arbitrage established in Stochastic Portfolio Theory (SPT), see Fernholz \cite{spt}, is of special interest to investment and portfolio management. However, market dynamics are constantly influenced by large investing entities where complicated interactions occur among them. We need a market model that captures these behaviors and develops a multi-agent optimization framework. To better describe and analyze the market based on SPT, this paper investigates the following questions: How do we capture the competitive behaviors of participants in the financial market? With additional information on these investors, how do we improve the market model and make portfolio suggestions? We aim to develop the optimization scheme for portfolio managers or asset management entities in a realistic market environment. This scheme would provide the information structure (for example, feedback from capitalization processes, wealth processes, or agent's preference profile, etc.) that is required for effective portfolio strategy and the corresponding optimal investments. 

The relative arbitrage problem first defined in SPT considers generating a strategy that outperforms a benchmark portfolio almost surely at the end of a certain time span and looks for the highest relative return. It is shown in \cite{diversity} that relative arbitrage can exist in equity markets that resemble actual markets and that relative arbitrage results from market diversity, a condition that prevents the concentration of all market capital into a single stock. Specific examples of the market, including the stabilized volatility model, in which a relative arbitrage opportunity exists, are introduced in \cite{ra2005}. To relax the assumptions about the behavior of the market imposed in the SPT, \cite{strong} considers relative arbitrage in regulated markets where dividends and the merging and splitting of companies are taken into account. 
Our model arises from the pioneering work of Fernholz and Karatzas \cite{opta}, which characterizes the best possible relative arbitrage with respect to the market portfolio, and derives non-anticipative investment strategies of the best arbitrage in a Markovian setting. The best arbitrage opportunity is further analyzed in \cite{uncertain} with  Knightian uncertainty. The smallest proportion of the initial market capitalization is described as the min-max value of a zero-sum stochastic game between the investor and the market. Further investigation of the exploitation of relative arbitrage opportunities has been carried out in \cite{visb, volar,ruf2, hedgear}. Assuming the market is diverse and sufficiently volatile, the functionally generated portfolios introduced in SPT are a tool to construct portfolios with favored return characteristics. 
The optimization problem from the functionally generated portfolio point of view is handled in \cite{tkwfgp}. The papers \cite{wong0} and \cite{wong} connect relative arbitrage with information theory and optimal transport problems. The robust optimization perspective is studied in \cite{itkin2022ergodic, itkin2021open, itkin2022robust} in terms of the asymptotic growth rate and model uncertainty.


{To address a gap in SPT regarding interactions among a large number of agents, this work considers the market dynamics influenced by the collective behavior of the agents. We develop a multi-agent optimization framework for relative arbitrage opportunities through stochastic differential games, in which investors seek optimal strategies to outperform the market index and their peers. Specifically, an investor $\ell$ achieves relative arbitrage if his/her logarithmic terminal wealth matches or exceeds the logarithmic terminal benchmark $\mathcal{V}(T)$ by a personalized preference level $c_{\ell}$ given at time 0. The benchmark is defined as a weighted average of market capitalization and average wealth.} 

%

\smallskip 


{The first question raised in this paper is: {\it Among noncooperative agents, under what conditions do relative arbitrage opportunities exist, and how can such opportunities be optimized?} 
We provide in Proposition~\ref{prop: clprop} a necessary condition under which every investor can achieve a relative arbitrage opportunity. Intuitively, this condition suggests that relative arbitrage is easier to attain when investors place greater weight on outperforming the market than on competing with their peers. Since the market also includes participants outside the group of investors considered here, competition against the market portfolio may imply potential arbitrage opportunities.  For example, the presence of large ``noise traders" can create such opportunities; see \cite{de1990noise} for a study of the market equilibrium with noise traders.
By contrast, if each investor seeks only to outperform the average wealth of the group, then no arbitrage opportunity arises. Another key determinant of the existence of relative arbitrage opportunities is the average preference, $\frac{1}{N}\sum_{\ell=1}^N c_\ell$, which reflects the overall intensity of competition to outperform the benchmark.
}

{
For investor $\ell$, optimal arbitrage is defined as the minimum initial wealth required to achieve relative arbitrage over the time horizon $[0,T]$. We study a group of sophisticated, competitive investors who seek optimal arbitrage with a common scheme specified later in the paper, whereas the drift and diffusion coefficients of the market dynamics are themselves affected by the actions of this group. In particular, for every $\ell$, we focus on the minimum ratio $u^\ell(T)$ (see Definition~\ref{uTdef}) such that the investor matches or exceeds $\mathcal{V}(T)$ at the terminal time when starting from the initial wealth $u^\ell(T)\mathcal{V}(0)$. Proposition~\ref{uniqpf} provides a sufficient condition under which investors \emph{strictly} outperform the benchmark. Our analysis is based on a PDE characterization of the optimal arbitrage quantity.
}

{The next question arises : {\it Is it possible for every investor to achieve optimal arbitrage in the market $\mathcal M$? If so, what are the corresponding strategies?} To analyze the interaction and competition among investors, we characterize the optimization problem as a differential game. The optimal strategies sought by investors in pursuit of relative arbitrage opportunities are captured by the Nash equilibrium. We then introduce a suitable notion of uniqueness for the Nash equilibrium. We derive the optimal strategy profile as a fixed-point problem over the path space of strategies, and we show that the Nash equilibrium exists.
In particular, Theorem~\ref{thm: neforxy} identifies a family of optimal strategies, inspired by functionally generated portfolios, that achieve Nash equilibrium. We show the uniqueness of the Nash equilibrium in Theorem~\ref{thm:unique}, where the time horizon and the scale of interactions must satisfy certain conditions. The strategy that achieves Nash equilibrium is in closed-loop feedback form, reflecting the dependence of the market structure on trading volumes. 
}

To conclude, this paper develops a multi-agent framework for relative arbitrage. The setup of \cite{opta} is recovered as a special case when
the number of investors $N =1$ and there is no interaction between the investor and the market.
{The results presented here motivate several directions for future research on stochastic portfolio theory via mean-field games, cooperative investors, short-term arbitrage opportunities. Beyond problems arising from stochastic portfolio theory, the characterization of interactions between agents and their environment through system dynamics, objective functionals, and feedback strategies is technically challenging but has wide-ranging applications in economics, engineering, and biology. We present a solution under such regime to decouple different types of interactions and search for Nash equilibrium. The framework developed here provides one approach to disentangling different forms of interaction and characterizing the associated Nash equilibrium. It also suggests further directions, including mean-field control, games under model uncertainty, and reinforcement learning. A detailed discussion of these topics is provided in Section~\ref{sec: future}.}

\vspace{3pt}

\textbf{Organization of this Paper.} This paper is structured as follows. Section \ref{fds} introduces the market with $N$ investors as a well-posed interacting particle system. Section \ref{chapu} discusses the relative arbitrage problem and the market price of risk processes in a multi-investor formulation. The conditions for obtaining relative arbitrage opportunities are explained in detail. In Section \ref{sec: nplayergame}, we set up $N$-player games among investors where the optimization of relative arbitrage is determined by Nash equilibrium. We discuss the existence and uniqueness of the Nash equilibrium and provide an example inspired by the volatility-stabilized market model. Finally, we include additional proofs and theoretical supports of the model in Appendices~\ref{defs}-\ref{racp}.

\section{The Market Model}
\label{fds}
\numberwithin{equation}{section}
We consider an equity market and focus on its aggregate behavior and a group of investors in this market. The group of investors under consideration is large enough to affect the market as a whole. However, this group of interest is only part of the entire market. The rest of the market participants do not influence the market and can be viewed as exogenous.

\subsection{Capitalizations, wealth and portfolios}

For a given finite time horizon $[0,T]$, an admissible market model $\mathcal{M}$ we use in this paper consists of a given $n$ dimensional standard Brownian motion $W(\cdot) := (W_1(\cdot), \ldots, W_n(\cdot))^{\prime}$ on the probability space $(\Omega, \mathcal{F}, \mathbb{P})$. 
The filtration $\mathbb{F}$ represents the ``flow of information'' in the market driven by the Brownian motion, that is, 
$\mathbb{F} = \{\mathcal{F}^W(t)\}_{0 \leq t < \infty}$
and $\mathcal{F}^W(t) := \{\sigma(W(s)) ; 0 < s < t\}_{0 \leq t < \infty} $ with $\mathcal{F}^W(0) := \{\emptyset, \Omega\}$, mod $\mathbb{P}$. All local martingales and supermartingales are with respect to the filtration $\mathbb{F}$ if not written specifically.

There are $n$ risky assets (stocks) with capitalizations $\mathcal{X}(\cdot) = (X_1(\cdot), \ldots, X_n(\cdot))^{\prime}$ driven by {$W(\cdot)$}. We define a factor $\mathcal Y (\cdot) := (\mathcal Y_1(\cdot), \ldots , \mathcal Y_n(\cdot))^\prime$, whose components represent the aggregate feedback effect of investors on the respective asset capitalizations. We will specify this term shortly. Here, $^\prime$ stands for the transpose of matrices. We assume $\beta(\cdot)$, $\sigma(\cdot)$, $\gamma(\cdot)$ and $\tau(\cdot)$ are defined on $\mathbb{R}_+^n \times \mathbb{R}_+^n$, are {time homogeneous} and the process $(\mathcal X(t), \mathcal Y(t)), t \ge 0  $ in Definition \ref{portfoliopi} is Markovian. That is, it follows the system of stochastic differential equations below: for $t \in [0,T]$, 
\begin{equation}
\label{eq: x}
dX_i(t) = X_i(t)(\beta_i(\mathcal{X}(t), \mathcal{Y}(t)) dt + \sum_{k=1}^n \sigma_{ik}(\mathcal{X}(t), \mathcal{Y}(t)) dW_k(t),\ \  i= 1, \ldots, n,
\end{equation}
with initial condition $X_i(0) = x_i$. The coefficients $\beta$ and $\sigma$ depend on the capitalization $\mathcal X$ and the factor $\mathcal Y$. Factor models have long been considered in the Capital Asset Pricing Models to assess the risk and return of the market. For example, in the Fama-French three factor model, the factors are market excess return, outperformance of small companies against big companies, and outperformance of high market-to-book ratio companies against low market-to-book ratio companies.    

Motivated by the continuous-time adjustment of equity capitalizations through supply and demand, we consider the average capital invested as a factor in the capitalization processes. Each investor $\ell$ invests the proportion $\pi_i^{\ell}(t)$ of current wealth $V^{\ell}(t)$ in the $i$th stock at each time $t$. We define $\mathcal Y(t) := (\mathcal Y_{1}(t), \ldots, \mathcal Y_{n}(t) )^\prime $, $t \ge 0 $. In the paper, 
 the factor process is chosen to represent interactions among investors. In particular, we consider the average trading volume 
\begin{equation}
\label{eq:defy}
    \mathcal Y_i(t) := \frac{1}{N} \sum_{\ell=1}^N V^{\ell}(t)  \pi_i^{\ell}(t), \quad \mathcal Y_i(0) := y_{0,i}, 
\end{equation}
where for each $i = 1,\ldots, N$, $\pi_i^{\ell}(\cdot)$ is fixed in Sections~\ref{fds}-\ref{chapu}, and satisfy the following Definition~\ref{portfoliopi} without concern for the interactions among investors. 
The average trading volume $\mathcal Y_i(t)$ invested by $N$ players in stock $i$ is assumed to follow an It\^o diffusion process 
\begin{equation} 
\label{eq: y}
\begin{aligned} 
   \mathcal Y_i(t)  &= y_{0,i} + \int^t_0 \gamma_i(\mathcal{X}(r), \mathcal{Y}(r)) d r + \int^t_0 \sum_{k=1}^n \tau_{ik}(\mathcal{X}(r), \mathcal{Y}(r)) d W_k(r), \ t \in (0,T] \text{ for an arbitrary } T > 0,\\
\end{aligned}
\end{equation}
and the initial value $\mathcal Y_i(0) := y_{0, i}$ 
for $i = 1, \ldots , n$. {In reality this is an interacting particle system constructed in Appendix C, and the coefficients $(\gamma, \tau)(\cdot)$ depends on the strategy, formulating a feedback loop. Indeed, in later sections, we address the market model to depend on the actions of investors, where there exist measurable functions $\phi^{\ell}$ such that $\pi_i^{\ell}(t) := \phi_i^{\ell}(\mathcal X(t), \mathcal Y(t))$ involves a fixed point problem in Section~\ref{sec: nplayergame}, for $i= 1, \ldots, n$, $\ell = 1, \ldots, N$. Then as $d \mathcal Y_i(t) = \frac{1}{N} \sum_{\ell=1}^N d \big[ V^{\ell}(t) \phi^{\ell}_i(\mathcal{X}(t), \mathcal{Y}(t)) \big]$, $i = 1, \ldots , n$, the coefficients $\gamma(\cdot)$ and $\tau(\cdot)$ are determined by It\^o's formula on $\phi^{\ell}(\cdot)$. Details can be found in Appendix~\ref{apdx:fp}. However, motivated by the mechanism of Mean Field Games and the equilibrium being the solution of fixed point problems, in Section 2-3, we discuss the static problem first, where the proportion $\pi$ is fixed and the average trading volume $\mathcal Y$ is considered to be a given process as in \eqref{eq: y}.
} 
 
 In this paper, we assume that $\text{dim}(W(t)) = \text{dim}(\mathcal{X}(t)) = n$, that is, we have exactly as many sources of randomness as there are stocks on the market $\mathcal{M}$. The dimension $n$ is chosen to be large enough to avoid unnecessary dependencies between the stocks we define. Here, $\beta(\cdot) = (\beta_1(\cdot), \ldots, \beta_n(\cdot))' : \mathbb{R}_+^n \times \mathbb{R}_+^n \rightarrow \mathbb{R}^n$ as the mean rates of return for $n$ stocks and $\sigma(\cdot) = (\sigma_{ik}(\cdot))_{n \times n} : \mathbb{R}_+^n \times \mathbb{R}_+^n \rightarrow \text{GL}(n)$, where $\text{GL}(n)$ is the space of $n \times n$ invertible real matrices. For simplicity, denote $\omega_t := (\mathcal{X}(t), \mathcal{Y}(t))$. $|\cdot|$ denotes the Euclidean norm for vector $\mathbb{R}^d$ and the Frobenius norm of matrix $\mathbb{R}^{d\times n}$, and in this paper $d =n$ or $N$. To satisfy the integrability condition, we assume that for any $T > 0$,
\begin{equation}
\label{xcond}
 \sum_{i=1}^n \int_0^T \bigg (|\beta_i(\omega_t)|+ \alpha_{ii}(\omega_t) \bigg)dt < \infty, 
\end{equation}
where $\alpha(\cdot) := \sigma(\cdot)\sigma'(\cdot)$, and its $i,j$ element $\alpha_{i,j}(\cdot)$ is the covariance process between the logarithms of $X_i$ and $X_{j}$ for $1\le i,j \le n$. The market $\mathcal{M}$ is hence a complete market. We assume $\gamma(\cdot)$ and $\tau(\cdot)$ satisfy that for any $T > 0$,
\begin{equation}
\label{tvcn}
    \sum_{i=1}^n \int_0^T \bigg (|\gamma_i(\omega_t)|+ \psi_{ii}(\omega_t) \bigg)dt < \infty,
\end{equation}
where $\psi(\cdot) := \tau(\cdot)\tau'(\cdot)$.

In this model, there are $N$ \textit{small} investors in the sense that each individual of these $N$ investors has very little influence on the overall system. 

\begin{definition}[Investment strategy]
\label{portfoliopi}
\begin{enumerate}[label={(\arabic*)}]
     \item 
An $\mathbb{F}$-progressively measurable $n$-dimensional process $\pi$ 
is called an {admissible} investment strategy { if $\pi$ is of Markovian feedback form, i.e., there exist measurable functions $\phi_t : \mathbb{R}^n \times \mathbb{R}^n \to \mathbb{R}^n$ such that  $\pi_t = \phi_t(\mathcal{X}_t, \mathcal{Y}_t)$, and if it satisfies}
\begin{equation}
\label{admpi}
\int_0^T (| (\pi(t))^{\prime}\beta(\omega_t)| + (\pi(t))^{\prime}\alpha(\omega_t)\pi(t) ) dt < \infty, \quad  \,  T\in (0,\infty),\,\, \text{a.e. } 
\end{equation}

The strategy here is a self-financing portfolio, since wealth at any time is obtained by trading the initial wealth according to the strategy $\pi(\cdot)$.
We denote the admissible set of the investment strategy process of an investor by $\mathbb{A}$. In the remainder of the paper, we only consider optimizing the strategy processes in the admissible set $\mathbb{A}$.

\item  
An investor uses the proportion $\pi_i(\cdot)$ of current wealth $V^{\ell}(\cdot)$ to invest in the stock $i$. The proportion $ \pi_0 = 1-\sum_{i=1}^n \pi_i(\cdot)$ is on the money market.  A special case of the admissible strategy is when 
 $\pi(\cdot) = (\pi_1(\cdot) , \ldots, \pi_n(\cdot) )'$
is a portfolio, i.e., it takes values in the set
$$\Delta_n := \{\pi = (\pi_1,...,\pi_n)\in \mathbb{R}^n \,| \pi_1 + \ldots +\pi_n = 1\}.$$
Denote $N$ investors' portfolios as $\Delta^N$, the Cartesian product of $\Delta_n$. From now on, we add the superscript $\ell$ to the strategy $\pi(\cdot)$ defined above to distinguish the strategies of different investors. The dynamics of the wealth process $V^{\ell}(\cdot)$ of an individual investor $\ell$, invested in the stock market, is determined by 
\begin{equation}
\label{wealth}
    \frac{dV^{\ell}(t)}{V^{\ell}(t)} = \sum_{i=1}^n \pi_i^{\ell}(t) \frac{dX_i(t)}{X_i(t)}, \quad  V^{\ell}(0) = v^{\ell}_{0}.
\end{equation}
\end{enumerate}
\end{definition}

In the rest of the paper, we simplify the notation as follows.
\[
X_i(t) \beta_i(\omega_t) = b_i(\mathcal{X}(t), \mathcal{Y}(t)),\] 
\[ X_i(t) \sigma_{ik} (t) = s_{ik}(\mathcal{X}(t), \mathcal{Y}(t)), \quad \sum_{k=1}^n s_{ik}(\omega_t)s_{jk}(\omega_t) = a_{ij}(\mathcal{X}(t), \mathcal{Y}(t)).
\]
\begin{asmp}
\label{asmp1:main}
\begin{enumerate}[label=\alph*.]
\item \label{asmp1alipfn}
Assume the Lipschitz continuity and linear growth condition are satisfied with Borel measurable mappings $b(\mathbf{x}, \mathbf{y})$, $s(\mathbf{x}, \mathbf{y})$, $\gamma(\mathbf{x}, \mathbf{y})$, $\tau(\mathbf{x}, \mathbf{y})$. For simplicity, we specify the conditions for $b(\mathbf{x}, \mathbf{y})$ and $s(\mathbf{x}, \mathbf{y})$ below, but the conditions for the coefficients $\gamma(\cdot)$ and $\tau(\cdot)$ of the trading volume processes $\mathcal{Y}(\cdot)$ can be written in the same vein.
That is, there exists a constant $C_L, C_G \in (0, \infty)$ that is independent of $t \in [0,T]$, such that
\begin{equation}
\label{bslip}
    |b(\mathbf{x}, \mathbf{y}) - b( \mathbf{x}', \mathbf{y}')| + |s(\mathbf{x}, \mathbf{y}) - s( \mathbf{x}', \mathbf{y}')| \leq C_L \left(|\mathbf{x}-\mathbf{x}'| + |\mathbf{y}-\mathbf{y}'|\right),
\end{equation}
\[
|b(\mathbf{x}, \mathbf{y})| + |s(\mathbf{x}, \mathbf{y})| \leq C_G(1+|\mathbf{x}|+ |\mathbf{y}|),
\]
for any $\left( \mathbf{x}, \mathbf{x}', \mathbf{y}, \mathbf{y}' \right) \in \mathbb{R}_+^n \times \mathbb{R}_+^n \times \mathbb{R}_+^n \times \mathbb{R}_+^n$.
\item \label{mktrisk} The matrix $\tau(\cdot) := \left( \tau_{ik}(\cdot) \right)_{1 \leq i,k \leq n}$ is nondegenerate. The market price of the risk process $\theta: \mathbb{R}^n \times \mathbb{R}^n \rightarrow \mathbb{R}^n$ exists and is square integrable. That is, there exists an $\mathbb{F}$-progressively measurable process such that for any $t \in [0, \infty)$,
\begin{equation}
\label{assum}
\sigma(\omega_t) \theta(\omega_t) = \beta(\omega_t), \ \ \tau(\omega_t) \theta(\omega_t) = \gamma(\omega_t),
\end{equation}
\[
\mathbb{P} \bigg( \int_0^T |\theta(\omega_t)|^2 dt < \infty, \forall T \in (0, \infty) \bigg) = 1,
\]
where $\omega_t := (\mathcal{X}(t), \mathcal{Y}(t))$.
\end{enumerate}
\end{asmp}

{In the remainder of the paper, we adopt the simplified notation $\theta(t) := \theta(\omega_t)$. As briefly introduced, the coefficients $\tau$ and $\gamma$ are not known a priori and need to be solved through a fixed point problem as laid out in Appendix~\ref{apdx:fp}. Thus, the Lipschitz conditions of $\tau$ and $\gamma$ in Assumption 1a) depend on the market coefficients $\beta$, $\sigma$ and strategies.} {The conditions on the market coefficients and strategies ensuring well-posedness of the associated interacting particle system are stated in Appendix~\ref{nmckean}.}
In the scope of a complete market, Assumption~\ref{asmp1:main}\ref{mktrisk} shows that the price of the risk process $\theta(t)$ governs both the risk premium per unit volatility of stocks and the trading volumes, since the market is simultaneously defined by the stocks and investors. The group of investors that we consider in this paper influences the stock through the trading volumes driven by the same $W(\cdot)$. Thus, it does not add an extra risk factor to the market.

\section{Optimization of relative arbitrage in finite systems}
\label{chapu}

We first recall the definition of relative arbitrage in Stochastic Portfolio Theory. 
\begin{definition}[Relative Arbitrage]
\label{ra}
Given two investment strategies $\pi(\cdot)$ and $\rho(\cdot)$, with the same initial capital $V^{\pi}(0) = V^{\rho}(0) =1$, we shall say that $\pi(\cdot)$ generates an arbitrage opportunity relative to $\rho(\cdot)$ over the time horizon $[0,T]$, with a given $T>0$, if
$$
\mathbb{P} \big( V^{\pi}(T) \geq V^{\rho}(T) \big) = 1 \quad \text{and} \quad \mathbb{P} \big(  V^{\pi}(T) > V^{\rho}(T) \big) >0. 
$$
\end{definition}

We use the total capitalization
\begin{equation} \label{Mportfolio-2}
X(t) = X_1(t) + \ldots + X_n(t), \quad t \in (0,T]; \quad X(0) = x_0 := x_1 + \cdots + x_n
\end{equation}
to represent the capitalization of the entire market.
The {\it market portfolio} $\mathbf{m}$ amounts to the ownership of the entire market by investing in proportion to the market weight of each stock,
\begin{equation}
\label{Mportfolio}
\mathbf{m}_i(t) = \pi^{\mathbf{m}}_i(t) := \frac{X_i(t)}{X(t)}, \quad i=1, \ldots, n, \quad t \ge 0.
\end{equation}
Consider the wealth process $V^{\mathbf{m}}(\cdot)$ generated by the market portfolio. Let $V^{\mathbf{m}}(0) = v_0$, 
\begin{equation}
\label{marketweal}
    \frac{d V^{\mathbf{m}}(t)}{V^{\mathbf{m}}(t)} = \sum_{i=1}^n \pi^{\mathbf{m}}_i(t) \cdot \frac{dX_i(t)}{X_i(t)} = \frac{d X(t)}{X(t)} , \quad t \ge 0. 
\end{equation}

In the following sections, instead of treating the reference strategy as the market portfolio as in \cite{opta}, we consider arbitrage opportunities relative to a modified benchmark strategy in Definition~\ref{defbenchmark} and Proposition~\ref{pistarr}. Without loss of generality, assume that $X(0) = V^{\mathbf{m}}(0)$ from now on.
\numberwithin{equation}{section}
\subsection{Relative arbitrage benchmark of the market and investors}

In general, the performance of a portfolio is measured with respect to the market portfolio and other factors. 
For example, asset managers care about not only absolute performance compared to the market index but also relative performance with respect to peer managers --- they try to exploit strategies that achieve an arbitrage relative to market and peer investors. Next, we define the benchmark for the overall performance. 
\begin{definition}[Benchmark]
\label{defbenchmark}
A relative arbitrage benchmark $\mathcal{V}(T)$, $T \in (0,\infty)$, which is the weighted average of the performance of the market portfolio and the average portfolio of the $N$ investors, is defined as
\begin{equation} 
\label{benchmark}
\mathcal{V}(T) = \delta \cdot X(T) + (1-\delta) \cdot  \overline{V}(T) , \quad \text{ where } \quad \overline{V}(\cdot) := \frac{1}{N} \sum_{\ell=1}^N V^{\ell}(\cdot) 
\end{equation}
is the average of the wealth processes $V^\ell(\cdot)$, $\ell=1,\ldots , N$.
$\delta \in (0,1]$ is a given constant weight. $X(T)$ is the wealth of the market portfolio in \eqref{Mportfolio-2}.
\end{definition}

We assume that each investor measures the logarithmic ratio of his or her own wealth at time $T$ to the benchmark in \eqref{benchmark}, and searches for a strategy such that the logarithmic ratio remains, almost surely, above a personal level of preference. For $\ell = 1, \ldots, N$, we denote the investment preference of investor $\ell$ by $c_{\ell}$, a real number given at $t=0$. Note that $c_{\ell}$ is an investor-specific constant and thus might be different between individuals $\ell = 1,\ldots, N$. An arbitrary investor $\ell$ tries to achieve 
\begin{equation} 
\label{eq: arbitrage3.2}
\log \frac{V^{\ell}(T)}{\mathcal{V}(T)} \ge c_{\ell}, \quad \text{a.s.} \quad 
\text{ or equivalently, } \quad 
V^{\ell}(T) \geq e^{c_{\ell}} \mathcal{V}(T), \quad \text{a.s.}
\end{equation}
Thus, $\mathcal{V}(T)$ is the benchmark, and an investor $\ell$ aims to match $e^{c_{\ell}}\mathcal{V}(T)$ according to their preferences. For simplicity, we define the average 
\[
\overline{e^{c}} := \frac{1}{N} \sum_{\ell=1}^N e^{c_{\ell}}.
\]

The following proposition shows that the benchmark $\mathcal{V}$ is a valid wealth process.
\begin{prop}
\label{pistarr}
The benchmark $\mathcal{V}(t) = \delta X(t) + (1-\delta) \overline{V}(T) $ in \eqref{benchmark} can be generated from a strategy $\Pi \in \mathbb{A}$, where $\Pi(\cdot) := (\Pi_1(\cdot), \ldots, \Pi_n(\cdot))$,
\begin{equation}
\label{bigpi}
\Pi_i(t) = \frac{1}{\mathcal{V}(t)} \big( \delta X_i(t) + (1-\delta)\mathcal{Y}_i(t)\big), 
\end{equation}
where $\mathcal{Y}_i(t)$ is defined in \eqref{eq: y}.
\end{prop}

\begin{proof}
To show that $\mathcal{V}(t)$ is a wealth process generated by a strategy, we use \eqref{wealth}, \eqref{eq: y}, \eqref{bigpi}  to get
\[
\frac{d\mathcal{V}(t)}{\mathcal{V}(t)} = \frac{1}{ \mathcal{V}(t)} \bigg(\delta dX(t) + \frac{1}{N}(1-\delta)\sum_{\ell=1}^N \sum_{i=1}^n V^{\ell} \pi_i^{\ell}(t) \frac{dX_i(t)}{X_i(t)} \bigg) = \sum_{i=1}^n \Pi_i(t) \frac{dX_i(t)}{X_i(t)} , \quad \text{for}  \ t \in (0,T],
\]
and
\begin{equation} \label{eq:mathcalV(0)}
\mathcal{V}(0) = \delta X(0) + \frac{1-\delta}{N}\sum_{\ell=1}^N v^{\ell},
\end{equation}
where
\[
\begin{aligned}
\Pi_i(t) &= \frac{\delta X(t)}{\mathcal{V}(t)} \mathbf{m}_i(t) + \frac{(1-\delta)}{N \mathcal{V}(t)} \sum_{\ell=1}^N V^{\ell} \pi^{\ell}_i(t) 
= \frac{\delta X_i(t) + (1-\delta)\mathcal{Y}_i(t)}{\mathcal{V}(t)}.
\end{aligned}
\]
Further computations show that $\Pi_i(t)$ satisfies the self-financing condition \eqref{admpi}. $\Pi \in \mathbb{A}$ since $\sum_{i=1}^N \Pi_i(t) = 1$ and $0 < \Pi_i(t) < 1$, for any $t \in [0,T]$, $i = 1, \ldots, n$. { For $i = 1, \ldots , n$, $t \ge 0$ we can write $\Pi$ as the feedback form \[\Pi_i(t) := \phi_i(\mathcal{X}(t), \mathcal{Y}(t)) = \frac{\delta X_i(t) + (1-\delta)\mathcal{Y}_i(t)}{\delta \sum_{i=1}^N X_i(t) + (1-\delta) \sum_{i=1}^N\mathcal{Y}_i(t)}.\]}
\end{proof}

\subsection{Optimization in relative arbitrage}
\label{sec: opt_in_ra}
Now, we shall answer the first question posed in the Introduction. What is the best strategy to achieve relative arbitrage over $[0,T]$ for investor $\ell = 1, \ldots , N$, with the fixed portfolios of the rest of the investors? {We adapt the idea in \cite{opta} and define the optimal arbitrage for an investor given that the portfolios of other investors are fixed. This is the first step in the construction of the optimal arbitrage between $N$ investors, since investors do not have access to the portfolios of other investors in reality.} 
Using the optimal strategy $\pi^{\ell \star}$, the investor $\ell$ will start with the least amount of initial capital (or initial cost) relative to $\mathcal{V}(0)$, to match or exceed the benchmark $e^{c_{\ell}} \mathcal{V}(T)$ at the terminal time $T$. {In the following definition, we add a superscript of ${v^{\ell}, \pi^{\ell}}$ to the wealth \eqref{wealth}, and denote it as $V^{v^{\ell}, \pi^{\ell}}(T)$ to emphasize that the wealth depends on the initial condition $v^{\ell}$ and the strategy $\pi^{\ell}(\cdot)$ in the below definition. For the rest of the paper, the notation of wealth is usually simplified as $V^\ell(\cdot)$.}

\begin{definition}[Optimal arbitrage among agents]
\label{uTdef}
Given the market dynamics with  {$(\mathcal{X}(0), \mathcal{Y}(0)) := (\mathbf{x}, \mathbf{y})$}, the initial wealth of the other investors $v^{-\ell}:= (v^{1}(\cdot),\ldots, v^{\ell-1}(\cdot), v^{\ell+1}(\cdot),\ldots,v^N(\cdot))$, and the admissible portfolios 
\begin{equation}
\label{eq: pi-ell}
    \pi^{-\ell}(\cdot) := (\pi^{1}(\cdot),\ldots, \pi^{\ell-1}(\cdot), \pi^{\ell+1}(\cdot),\ldots,\pi^N(\cdot)),
\end{equation}
where each portfolio $\pi^{k} \in \mathbb{A}$, $k = 1, \ldots, \ell-1, \ell+1, \ldots, N$.
In the market system \eqref{eq: x}-\eqref{eq: y}, for every $\ell = 1, \ldots, N$, the investor $\ell$ pursues the optimal arbitrage characterized by the smallest initial relative wealth
\begin{equation}
\label{utobj}
u^{\ell}(T) = \inf \bigg \{ w^{\ell} \in (0, \infty) \,  \Big \vert \, \text{ there exists } 
\pi^{\ell} \in \mathbb{A} \text{ such that } \,v^{\ell} = w^{\ell} \mathcal{V}(0) , \, \,  V^{v^{\ell}, \pi^{\ell}}(T) \geq e^{c_{\ell}} \cdot \mathcal{V}(T) \bigg \} 
\end{equation}
and the relative arbitrage portfolio $\{\pi^\ell(t)\}_{0 \leq t \leq T}$ that achieves the smallest initial relative wealth $u^\ell(T)$. 
\end{definition}

Note that in Definition~\ref{uTdef}, the optimal strategy vector $\pi^{\ell}(\cdot)$ is not necessarily uniquely determined. Thus, the investor $\ell$ searches for his or her smallest initial wealth relative to the benchmark $\mathcal{V}(0)$, and the optimal strategy $\pi^{\ell}(\cdot)$ as a pair.  The initial benchmark $\mathcal{V}(0)$ is fixed {as $(\mathbf{x}, \mathbf{y})$ is given in the definition, and the smallest initial wealth of the investor $\ell$ follows $v^\ell = \omega^\ell \mathcal{V}(0) = \omega^\ell \left(\delta \mathbf{x} \cdot \mathbf{1} + (1-\delta) \mathbf{y}\cdot \mathbf{1}\right)$. In other words, although we have the definition \eqref{eq:defy} and \eqref{benchmark}, here $v^{-\ell}$ and $\pi^{-\ell}(\cdot)$ are fixed, and $v^\ell$ is given by $(\mathbf{x}, \mathbf{y})$ for the above reasons.}

We define the deflator based on the market price of the risk process $L(t)$ as
\begin{equation}
    \label{ltmg}
dL(t) = - \theta(\omega_t) L(t) dW_t, \quad t \ge 0, \quad L(0) = 1 .  
\end{equation}
{Recall that $\omega_t := (\mathcal{X}(t), \mathcal{Y}(t))$ as in Assumption~\ref{asmp1:main}.} Equivalently,
\[
L(t) := \exp \Big\{- \int_0^t \theta'(\omega_s) dW(s) - \frac{1}{2} \int_0^t |\theta(\omega_s)|^2 ds \Big\}, \quad 0 \leq t < \infty. 
\]

The market is endowed with the existence of a local martingale $L$ with $\E[L(T)] \leq 1$ under Assumption~\ref{asmp1:main}\ref{mktrisk}. The discounted processes $\widehat{V}^{\ell}(\cdot) := V^{\ell}(\cdot) L(\cdot)$, and $\widehat{X}(\cdot) := X(\cdot) L(\cdot)$. $\widehat{V}^{\ell}(\cdot)$ admits
\begin{equation}
\label{vlhatt}
\begin{aligned}
   d \widehat{V}^{\ell}(t) = d V^{\ell}(t) L(t) =  \widehat{V}^{\ell}(t)\big(\pi^{\ell\prime}(t) \sigma(t) -  \theta'(\omega_t) \big)dW(t); \quad  \widehat{V}^{\ell}(0) = \widehat{v}_{\ell}, \, \, \ell = 1, \ldots , N\, , 
\end{aligned}
\end{equation}
\begin{equation}
\label{xihat}
d \widehat{X}_i(t) = \widehat{X}_i(t) \sum_{k=1}^n (\sigma_{ik}(t) - \theta_k(\omega_t)) dW_k(t); \quad \widehat{X}_i(0) = x_i, \, \, i =1, \ldots , n \, ,  
\end{equation}
\begin{equation}
\label{xhat}
d \widehat{X}(t) = \widehat{X}(t) \sum_{k=1}^n (\sum_{i=1}^n \mathbf{m}_i(t) \sigma_{ik}(t) - \theta_k(\omega_t)) dW_k(t); \quad \widehat{X}(0) = x.
\end{equation}

\begin{remark}
\label{clprop1}
Under Assumption~\ref{asmp1:main}\ref{mktrisk}, suppose, in addition, that the market $\mathcal{M}$ has bounded variance (Appendix~\ref{defs}). Denote $\Bar{v} := \frac{1}{N} \sum_{\ell=1}^N v^{\ell}$. On $[0,T]$, given the existence of relative arbitrage in the sense of \eqref{eq: arbitrage3.2} and Definition~\ref{defbenchmark},  the process $L(\cdot)$ is a strict local martingale, i.e., $\mathbb{E}[L(T)]<1$.

This can be proved by contradiction, assuming that $L(T)$ is a martingale. We generalize Proposition 6.1 in \cite{mf} where the case $N=1$ is studied. 
As $u^{\ell}(T) = \frac{v^{\ell}}{\mathcal{V}(0)} \leq e^{c_{\ell}}$ for some $\ell$, it holds that
\begin{equation}
\label{clieq}
 c_{\ell} \geq \log v_{\ell} - \log(\delta x_0 + (1-\delta)\Bar{v}). 
\end{equation} 
By the Girsanov theorem, $\mathbb{Q}_{\mathbf{x}, \mathbf{y}}(A) := \mathbb{E}^{\mathbf{x}, \mathbf{y}}[L(T) 1_A], A \in \mathcal{F}_T$ defines a probability measure that is equivalent to $\mathbb{P}$. We can show that $\Delta^{\ell}(t) := V^{\ell}(t) - e^{c_{\ell}}(\delta X(t) + (1-\delta) \frac{1}{N} \sum_{k=1}^N V^{k}(t))$ is a martingale under $\mathbb{Q}_{\mathbf{x}, \mathbf{y}}$. 
Hence, the martingale property of $\Delta^{\ell}(t)$ and \eqref{clieq} suggests that 
$ \mathbb{E}^{\mathbb{Q}_{\mathbf{x}, \mathbf{y}}} [\Delta^{\ell}(T)] = \mathbb{E}^{\mathbb{Q}_{\mathbf{x}, \mathbf{y}}} [\Delta^{\ell}(0)] = v_{\ell}-e^{c_{\ell}}\delta x_0 - e^{c_{\ell}}(1-\delta)\Bar{v} \leq 0$, 
contrary to Definition~\ref{ra} of the relative arbitrage. Thus, the process $L(\cdot)$ is a strict local martingale.
\end{remark}

\subsection{The modified subproblems}
\label{dp}


{We emphasize that in Section~\ref{chapu}, $(\mathbf{x}, \mathbf{y})$ are fixed and the dynamics are assumed to be given. In reality, the investor has no access to the state and strategy of everyone else; otherwise, this would be the same as the single investor's optimal arbitrage \cite{opta} with a modified market dynamics. Nevertheless, the previous analysis gives a first step in decoupling the trajectories of market dynamics, investor wealth, and investment strategies so that the problem can be analyzed first from a single investor's perspective. We continue this single investor's perspective in this section and characterize the optimal solution to relative arbitrage opportunities. The complete setup will be treated in Section~\ref{sec: nplayergame}.}


The following proposition \ref{f1} characterizes the best relative arbitrage opportunities using the  benchmark $e^{c_{\ell}} \cdot \mathcal{V}(T)$, for any $\ell = 1, \ldots, N$. $T$ is a fixed real positive number and $N$ is a fixed natural number. 
We shall assume that $\mathbb{F} = \mathbb{F}^{\mathcal{X},\mathcal{Y}} = \mathbb{F}^W$, where $\mathbb{F}^{\mathcal{X},\mathcal{Y}}$ is the filtration generated by the $\sigma$-fields $\{\sigma(\mathcal{X}(s), \mathcal{Y}(s) ; 0 < s < t), t \ge 0 \}$. In general, the trading strategy does not necessarily have to be measurable with respect to $\mathbb{F}^{\mathcal{X},\mathcal{Y}}$. An example of a single investor case is presented in \cite[Example 6]{ruf2}. In Section~\ref{sec: nplayergame}, the trading strategies generated in $\mathbb{F}^{\mathcal{X},\mathcal{Y}}$ are the optimal strategy that we will derive explicitly. Then $\mathbb{F}^{W} = \mathbb{F}^{\mathcal{X},\mathcal{Y}}$ is a natural relation from the structure of the optimal strategy we consider and is required for the derivation of the results of the martingale representation below. {In the following, $\mathbb{E}^{\mathbf{x}, \mathbf{y}} [\cdot] := \mathbb{E} \left[\cdot| \left(\mathcal{X}(0), \mathcal{Y}(0) \right) = (\mathbf{x}, \mathbf{y}) \right]$.}

\begin{prop}
\label{f1}
Under Assumption~\ref{asmp1:main}, for $\ell = 1, \ldots, N$, assume {$\pi^{\ell}(s)$ to be the optimal strategy over $s \in [t,T]$ in Definition~\ref{uTdef} that achieves $u^{\ell}(T-t, \mathbf{x}, \mathbf{y})\in C^{1,3,3}$}. With the admissible portfolio $\pi^{-\ell}(\cdot)$ defined in \eqref{eq: pi-ell}, and the initial values $(\mathbf{x}, \mathbf{y})$, 
$u^{\ell}(T)$ in \eqref{utobj} can be derived as the discounted expected values of $e^{c_{\ell}} \mathcal{V}(T)$, that is,
\begin{equation}
\label{ute}
        u^{\ell}(T, \mathbf{x}, \mathbf{y}) = e^{c_{\ell}} \mathbb{E}^{\mathbf{x}, \mathbf{y}} \big[ \mathcal{V}(T) L(T) \big]\, /\, \mathcal{V}(0).
\end{equation}
\end{prop}

This result is essential for the PDE characterization of the objective $u^{\ell}(\cdot)$ in Section~\ref{5.1}. {$u^{\ell}(\cdot)$ is an investor-specific quantity through the preference term $e^{c_{\ell}}$. For an arbitrary individual $\ell$, $c_{\ell}$ is determined at initial time, and \eqref{ute} is satisfied for this individual. The choice of $u^\cdot(t, \mathbf{x}, \mathbf{y}) \in C^{1,3,3}((0, \infty) \times (0, \infty)^n \times (0, \infty)^n)$ guarantees that the optimal strategy $\pi^{\ell}$ is in the admissible set $\mathbb{A}$, which is explained later in Appendix D, equation \eqref{gonc}.}
The proof is based on the supermartingale property of $\widehat{V}^{\ell}(\cdot)$ and the martingale representation theorem; see Appendix~\ref{racp} for details of the proof.

Denote $\tilde \tau := T-t$. By Assumption~\ref{asmp1:main}\ref{mktrisk} and the Markovian market setup, we have the Markovian property of $u^{\ell}(\cdot)$, 
\begin{equation}
\label{ggg}
   u^{\ell}(\tilde \tau, \mathbf{x}, \mathbf{y})= e^{c_{\ell}} \frac{\mathbb{E}[\mathcal{V}(T)L(T) | \mathcal{F}_t] }{\mathcal{V}(t)L(t)}.
\end{equation}
We can understand $u^{\ell}(T-t)$ as the initial optimal arbitrage quantity to start at $t \in [0,T)$ such that we match or exceed the benchmark portfolio at the terminal time $T$, that is,
\begin{equation}
\label{utobj1}
u^{\ell}(T-t) = \inf \bigg \{ \omega^{\ell} \in (0, \infty) \,  \Big \vert \, \text{ there exists } 
\widetilde{\pi}^{\ell} \in \mathbb{A} \text{ such that } \,v^{\ell} = \omega^{\ell} \widetilde{\mathcal{V}}(t) , \, \,  \widetilde{V}^{{\ell}}(T) \geq e^{c_{\ell}} \cdot \widetilde{\mathcal{V}}(T) \bigg \} 
\end{equation}
for $ 0 \le t \le T$ and $\ell = 1, \ldots , N$. 

At every time $t$, each investor 
optimizes $\, \widetilde{\pi}^\ell(\cdot)\, $ from $t$ to $T$, to obtain the optimal quantity as defined in \eqref{utobj1}. In other words, here $\widetilde{V}^\ell(T) = \widetilde{V}^{v^\ell , \widetilde{\pi}^\ell} (T) \,$ is generated by the admissible portfolio $\, \{ \widetilde{\pi}^\ell (s) , s \ge t\} \, $ starting from time $t \ge 0$. That is, we consider the portfolio $\, \widetilde{\pi}^\ell(s)\, $, $\, t \le s \le T$ so that the corresponding portfolio value $\, \widetilde{V}^\ell (s)\, $, $\, t \le s \le T\, $ satisfies $\, \widetilde{V}^\ell (T) \ge e^{c_\ell} \cdot \widetilde{\mathcal{V}}(T)$, where $\widetilde{V}^{\ell} (t) = v^\ell = \omega^\ell \widetilde{\mathcal{V}}(t) \, ,$
\[
\frac{d \widetilde{V}^{\ell}(s)}{ \widetilde{V}^{\ell} (s)} = \sum_{i=1}^n \widetilde{\pi}_i^\ell (s) \cdot \frac{d X_i (s)}{\, X_i (s) \,}\, ; \quad t \le s \le T\, ,  
\]
\[
\widetilde{\mathcal{V}}(s) := \delta \cdot X(s) + (1-\delta) \cdot \frac{1}{N} \sum_{\ell=1}^N \widetilde{V}^{\ell}(s), \quad t \le s \le T \, .
\]

{The following proposition gives the conditions for $\{c_{\ell}\}$ and $\delta$ to achieve relative arbitrage opportunities. It also gives the initial amount of capital required for homogeneous investors. We show that direct interaction with the market is important.
One possible interpretation of relative arbitrage here is the presence of noise traders whose correlated misperceptions of returns generate additional fluctuations in the market. Consistent with the observation in \cite{de1990noise}, the presence of noise traders offers institutional investors a better opportunity, in the sense that a market with only institutional investors will not have arbitrage opportunities of risky assets.}
\begin{prop}
\label{prop: clprop}
 {Given the initial conditions $(\mathbf{x}, \mathbf{y}) \in \mathbb{R}_+^n \times \mathbb{R}_+^n$, the following properties of relative arbitrage opportunities with respect to $c_{\ell}$ and $\delta$ hold:}
\begin{enumerate}
    \item If every investor achieves a relative arbitrage opportunity in the sense of \eqref{eq: arbitrage3.2}, then we must have 
    \begin{equation}
    \label{ccondc}
    \frac{(1-\delta)}{N} \sum_{\ell=1}^N e^{c_{\ell}} < 1.       
    \end{equation}
    
\item There is no relative arbitrage opportunity when $\delta=0$ and $c_{\ell}\geq 0$ for every $\ell$.

\item When $0<\delta \leq 1$,
we have
\begin{equation}
\label{v0v0}
    \mathcal{V}(0) = \frac{\delta x_0}{1-\frac{1-\delta}{N} \sum_{k=1}^N u^{k}(T, \mathbf{x}, \mathbf{y})}, \quad \text{thus} \quad v^{\ell} = \frac{ u^{\ell}(T, \mathbf{x}, \mathbf{y}) \delta x_0}{1-\frac{1-\delta}{N} \sum_{k=1}^N u^{k}(T, \mathbf{x}, \mathbf{y})},
\end{equation}
where $x_0$ is the sum of each coordinate of $\mathbf{x}$, as defined in \eqref{Mportfolio-2}.

\end{enumerate}
\end{prop}

Thus, in this paper, we mainly discuss relative arbitrage opportunities with $0<\delta \leq 1$. We now provide the proof and discuss its financial interpretation.

\begin{proof}
\begin{enumerate}
    \item Since everyone follows $V^{\ell}(T) \geq e^{c_{\ell}} \mathcal{V}(T)$, we sum up this expression for $\ell = 1, \ldots, N$ to get an inequality of $\sum_{\ell=1}^N V^{\ell}(t)/N$, and \eqref{ccondc} follows immediately. We easily see that if all investors adopt the market portfolio $\mathbf{m}(\cdot)$, then any $c_{\ell} \leq 0$, $\ell = 1, \ldots, N$ is a valid level of satisfaction. Here, \eqref{ccondc} 
tells us that some of $c_{\ell}$'s can be small positive numbers. 
\item Next, by Definition~\ref{ra}, if we have 
\[
c_{\ell} \leq  \log \bigg ( \frac{V^{\ell}(T)}{\mathcal V^{N}(T)} \bigg ) = \log \bigg(\frac{V^{\ell}(T)}{\delta X^N(t) + (1-\delta) \overline{V}(T) }\bigg) , \quad \ell = 1, \ldots , N , 
\] 
then the relative arbitrage exists in the sense of \eqref{eq: arbitrage3.2}. 

From \eqref{utobj}, we get $v^{\ell} = u^{\ell}(T) \mathcal{V}(0)$ for $\ell = 1, \ldots, N$. When $\delta = 0$, $c_{\ell} \geq 0$ for every $\ell$, then \eqref{ccondc} is violated, and there is no relative arbitrage opportunity. This means that in an extremely competitive group, where every investor solely wants to outperform one another in the market, the relative arbitrage is not possible. 

\item Let $\mathbf{1}$ be the $n \times 1$ column vector that has all $n$ elements equal to one. When $\delta>0$, by Definition~\ref{uTdef},
\begin{equation}
\label{eq: fx-u0}
u^{\ell}(T, \mathbf{x}, \mathbf{y}) \mathcal{V}(0) = v_0^{\ell}, \quad (\mathbf{x}, \mathbf{y}) :=(\mathcal{X}(0), \mathcal{Y}(0)),
\end{equation}
where
\begin{equation}
    \label{calv}
\mathcal{V}(0) = \delta \sum_{i=1}^n x_i + \sum_{i=1}^n \sum_{\ell=1}^N v^{\ell}_0 \pi^{\ell}_i(0) = \delta \sum_{i=1}^n x_i + (1-\delta) \sum_{i=1}^n y_i = \delta \mathbf{x} \cdot \mathbf{1} + (1-\delta) \mathbf{y} \cdot\mathbf{1}.
\end{equation}
By \eqref{eq: fx-u0}, it follows that $\mathbf{y} \cdot\mathbf{1} = \frac{1}{N} \sum_{\ell=1}^N u^{\ell}(T, \mathbf{x}, \mathbf{y}) [\delta \mathbf{x} \cdot \mathbf{1} + (1-\delta) \mathbf{y} \cdot\mathbf{1}]$, so
\begin{equation}
    \label{eq:ysumfix}
    \mathbf{y} \cdot\mathbf{1} =  \frac{1}{N}  \delta \mathbf{x} \cdot \mathbf{1}\sum_{\ell=1}^N u^{\ell}(T, \mathbf{x}, \mathbf{y}) \left[ 1-(1-\delta)\frac{1}{N} \sum_{\ell=1}^N u^{\ell}(T, \mathbf{x}, \mathbf{y}) \right]^{-1}.
\end{equation}
Further computation based on \eqref{eq: fx-u0}-\eqref{eq:ysumfix} shows that \eqref{v0v0} follows. 

\end{enumerate}
\end{proof}

\subsection{PDE characterization of the best relative arbitrage}
\label{5.1}

In this section, we characterize $u^{\ell}(\tilde \tau,  \mathbf{x}, \mathbf{y})$ in \eqref{ggg} as a solution to a Cauchy problem. 

We write $D_i$ and  $D_{ij}$  for  the first- and second-order partial derivatives  with  respect  to  the $i$-th or the $(i,j)$-th variables in $\mathcal{X}(\cdot)$, respectively; $D_p$ and  $D_{pq}$  for the first and second partial derivative in $\mathcal{Y}(\cdot)$, $D_{pi}$ is the cross derivative in $\mathcal{X}(\cdot)$ and $\mathcal{Y}(\cdot)$.

\begin{asmp}
\label{hasmp}
There exists a function $H: \mathbb{R}_+^n \times \mathbb{R}_+^n \rightarrow \mathbb{R}$ of class $C^2$, such that,
\[
b( \mathbf{x}, \mathbf{y}) = 2a( \mathbf{x}, \mathbf{y}) D_x H( \mathbf{x},\mathbf{y}), \quad \gamma( \mathbf{x}, \mathbf{y}) = 2\psi( \mathbf{x}, \mathbf{y}) D_y H(\mathbf{x}, \mathbf{y}).
\]
That is, component-wise $b
_i(\cdot) =\sum_{j=1}^n a_{ij}(\cdot) D_{j} H (\cdot)$,  $\gamma
_p(\cdot) =\sum_{q=1}^n \psi_{pq}(\cdot) D_{q} H(\cdot)$, for $i,p = 1, \ldots , n$. 
\end{asmp}
{We provide an example that satisfies the assumption above in Example~\ref{ex: vsm}, which is inspired by volatility-stabilized market models.}


The generator for the process $(\mathcal X(\cdot), \mathcal Y(\cdot)) $ can be written as
\[
\begin{aligned}
 \mathcal{L}f & := \sum_{i=1}^n \sum_{j=1}^n a_{ij}( \mathbf{x}, \mathbf{y}) \big[ \frac{1}{2}  D_{ij} f + 2 D_i f D_j H( \mathbf{x}, \mathbf{y}) \big] + \sum_{p=1}^n \sum_{q=1}^n \psi_{pq}( \mathbf{x}, \mathbf{y}) \big[ \frac{1}{2}  D_{pq} f + 2 D_p f D_q H( \mathbf{x}, \mathbf{y}) \big]\\
 &+ \frac{1}{2} \sum_{i=1}^n \sum_{p=1}^n (s \tau')_{ip}( \mathbf{x}, \mathbf{y}) D_{ip} f + \frac{1}{2} \sum_{i=1}^n \sum_{p=1}^n (\tau s')_{pi}( \mathbf{x}, \mathbf{y}) D_{pi} f,
 \end{aligned}
\]
where $(\tau s')_{pi}( \mathbf{x}, \mathbf{y}) = (s \tau')_{ip}( \mathbf{x}, \mathbf{y}) = \sum_{k=1}^K s_{ik}(\mathbf{x}, \mathbf{y}) \tau_{pk}( \mathbf{x}, \mathbf{y})$. By the definition of $\theta(\cdot)$ in \eqref{assum} and Assumption~\ref{hasmp}, 
 \begin{equation}
 \label{cl1}
     \theta( \mathbf{x}, \mathbf{y}) = s'( \mathbf{x}, \mathbf{y}) D_x H(\mathbf{x}, \mathbf{y}) + \tau'( \mathbf{x}, \mathbf{y}) D_y H(\mathbf{x}, \mathbf{y}).
 \end{equation}
Then it follows from \eqref{cl1} and It\^o's lemma applying on $H(\cdot)$ that
\[
\begin{aligned}
    L(t) &= \exp \bigg \{-\int_0^t \theta'(s) dW(s) - \frac{1}{2} \int_0^t |\theta(s)|^2 ds \bigg \}\\
    &= \exp \bigg\{ - H(\mathcal{X}(t),\mathcal{Y}(t)) +  H( \mathbf{x},\mathbf{y}) - \int_0^{t} ( k(\mathcal{X}(s), \mathcal{Y}(s)) + \Tilde{k}(\mathcal{X}(s), \mathcal{Y}(s)) ) ds \bigg\},
\end{aligned}
\]
where
\begin{equation}
\label{eq: kxy}
k(\mathbf{x}, \mathbf{y}) := - \sum_{i=1}^n \sum_{j=1}^n \frac{a_{ij}( \mathbf{x}, \mathbf{y})}{2} [ D_{ij}^2 H( \mathbf{x},\mathbf{y}) + 3 D_i H ( \mathbf{x},\mathbf{y}) D_j H( \mathbf{x},\mathbf{y})],    
\end{equation}
\begin{equation}
\label{eq: ktildexy}
\begin{aligned}
\Tilde{k}( \mathbf{x}, \mathbf{y}) := -\sum_{p=1}^n \sum_{q=1}^n \frac{\psi_{pq}( \mathbf{x}, \mathbf{y})}{2} [ D_{pq}^2 H(\mathbf{x},\mathbf{y}) + 3 D_p H(\mathbf{x},\mathbf{y}) D_q H(\mathbf{x},\mathbf{y})] + \sum_{i=1}^n \sum_{p=1}^n (s \tau')_{ip} D_i H ( \mathbf{x},\mathbf{y}) D_p H( \mathbf{x},\mathbf{y}) 
\end{aligned}
\end{equation}
for $( \mathbf{x},\mathbf{y}) \in (0,\infty)^n \times (0,\infty)^n$. 
Thus, \eqref{ggg} can be written as 
\begin{equation}
\label{ggg1}
   u^{\ell}(\tilde \tau, \mathbf{x}, \mathbf{y})= e^{c_{\ell}} \frac{G(\tilde \tau, \mathbf{x}, \mathbf{y})}{g( \mathbf{x}, \mathbf{y})},
\end{equation}
where, by \eqref{calv},
\begin{equation}
    \label{geqng}
g( \mathbf{x}, \mathbf{y}) := \mathcal{V}(0) e^{-H( \mathbf{x},\mathbf{y})} = \big( \delta \mathbf{x} \cdot \mathbf{1} + (1-\delta) \mathbf{y} \cdot\mathbf{1} \big) e^{-H( \mathbf{x},\mathbf{y})},
\end{equation}
\[
G(T, \mathbf{x}, \mathbf{y}) :=  \mathbb{E}^{\mathbf{x}, \mathbf{y}} \big[ g(\mathcal{X}(T), \mathcal{Y}(T)) e^{- \int_0^{T} k(\mathcal{X}(t))+ \Tilde{k}(\mathcal{Y}(t)) dt} \big]. 
\]

\begin{asmp}
\label{asmp: geqn}
The function $G(\cdot) \in C^{1,3,3}\left((0,\infty) \times (0, \infty)^n \times (0, \infty)^n \right)$ yields the following dynamics by the Feynman-Kac formula,
\begin{equation}
\begin{aligned}
    \frac{\partial G}{\partial \tilde{\tau}}(\tilde \tau, \mathbf{x}, \mathbf{y}) = \mathcal{L} G(\tilde \tau, \mathbf{x}, \mathbf{y})- \big( &k(\mathbf{x}, \mathbf{y})+ \Tilde{k}(\mathbf{x}, \mathbf{y}) \big) G(\tilde \tau,  \mathbf{x}, \mathbf{y}), \ \ \ \ (\tilde \tau, \mathbf{x}, \mathbf{y})\in \mathbb{R}_+ \times \mathbb{R}_+^n \times \mathbb{R}_+^n,\\
G(0, \mathbf{x}, \mathbf{y}) = &g( \mathbf{x}, \mathbf{y}),\  ( \mathbf{x}, \mathbf{y}) \in \mathbb{R}_+^n \times \mathbb{R}_+^n,
\end{aligned}
\end{equation}
\end{asmp}

By \cite{bmk}, this assumption is satisfied if we also assume that $k(\cdot)$ is $\alpha$-H\"older continuous with some $\alpha \in (0,1]$, uniformly in compact subsets of $\mathbb{R}_+^n \times \mathbb{R}_+^n$, $\ell = 1,\ldots, N$; $k(\mathbf{x},\mathbf{y})$ is bounded from below, $\gamma(\cdot)$ and $\tau(\cdot)$ are continuously differentiable on $(0, \infty)^n \times (0, \infty)^n$, and satisfy the growth condition 
\[
|\gamma(\mathbf{x},\mathbf{y})| + | \tau(\mathbf{x},\mathbf{y})| \leq M (1 + |\mathbf{x}| + |\mathbf{y}|).
\]
Under Assumption~\ref{asmp: geqn},
$u^{\ell}(\tilde \tau, \mathbf{x}, \mathbf{y})\in C^{1,3,3}((0, \infty) \times (0, \infty)^n \times (0, \infty)^n)$ is bounded on $K \times (0, \infty)^n \times (0, \infty)^n$ for each compact $K \subset (0, \infty)$. After the calculations (Appendix~\ref{racp}) based on \eqref{ggg} and \eqref{geqng}, $u^{\ell}(\cdot)$ follows a Cauchy problem
\begin{equation}
\label{eqsol}
\frac{\partial u^{\ell}(\tilde \tau,  \mathbf{x}, \mathbf{y})}{\partial \tilde \tau} = \mathcal{A} u^{\ell}(\tilde \tau,  \mathbf{x}, \mathbf{y}), \quad \tilde \tau \in (0,\infty),\ ( \mathbf{x}, \mathbf{y}) \in (0, \infty)^n \times (0, \infty)^n,
\end{equation}
\begin{equation}
    \label{eqsol2}
u^{\ell}(0,  \mathbf{x}, \mathbf{y}) = e^{c_{\ell}}, \hspace{1.5cm} ( \mathbf{x}, \mathbf{y}) \in (0, \infty)^n \times (0, \infty)^n,
\end{equation}
where
\begin{equation}
\label{aineq}
\begin{aligned}
\mathcal{A} u^{\ell}(\tilde \tau, \mathbf{x}, \mathbf{y})=& \frac{1}{2} \sum_{i=1}^n \sum_{j=1}^n a_{ij} ( \mathbf{x}, \mathbf{y}) \Big(D_{ij}^2 u^{\ell}(\tilde \tau, \mathbf{x}, \mathbf{y})+ \frac{ 2 \delta D_i u^{\ell}(\tilde \tau,  \mathbf{x}, \mathbf{y})}{ \delta \mathbf{x} \cdot \mathbf{1} + (1-\delta) \mathbf{y} \cdot\mathbf{1}}\Big) \\
& + \frac{1}{2} \sum_{p=1}^n \sum_{q=1}^n \psi_{pq} ( \mathbf{x}, \mathbf{y}) \Big(D_{pq}^2 u^{\ell}(\tilde \tau,  \mathbf{x}, \mathbf{y})+ \frac{ 2 (1-\delta) D_p u^{\ell}(\tilde \tau,  \mathbf{x}, \mathbf{y})}{ \delta \mathbf{x} \cdot \mathbf{1} + (1-\delta) \mathbf{y} \cdot\mathbf{1}}\Big)\\
& + \sum_{i=1}^n \sum_{p=1}^n (s \tau^T)_{ip}( \mathbf{x}, \mathbf{y}) D_{ip}^2 u^{\ell}(\tilde \tau,  \mathbf{x}, \mathbf{y})\\
& + \sum_{i=1}^n \sum_{p=1}^n (s \tau^T)_{ip}( \mathbf{x}, \mathbf{y}) \frac{ \delta D_p u^{\ell}(\tilde \tau, \mathbf{x}, \mathbf{y})+ (1-\delta) D_i u^{\ell}(\tilde \tau,  \mathbf{x}, \mathbf{y})}{ \delta \mathbf{x} \cdot \mathbf{1} + (1-\delta) \mathbf{y} \cdot\mathbf{1}}.
\end{aligned}
\end{equation}
We emphasize that \eqref{eqsol} is entirely determined by the volatility structure of $\mathcal{X}(\cdot)$ and $\mathcal{Y}(\cdot)$. As a result, when the drift term $\gamma(\cdot)$ and the volatility term $\tau(\cdot)$ in \eqref{eq: y} are given, \eqref{eqsol} - \eqref{aineq} are satisfied. 

\begin{thm}
\label{thm31}
Under Assumption~\ref{asmp1:main}-\ref{asmp: geqn}, the function $u^{\ell} : [0,\infty) \times (0,\infty)^n \times (0,\infty)^n \rightarrow (0, 1]$ is the smallest nonnegative continuous function, of class $C^{1,3,3}$ on $(0,\infty) \times (0,\infty)^n \times (0,\infty)^n$, which satisfies $u^{\ell}(0, \cdot) \equiv e^{c_{\ell}}$ and
\begin{equation}
\label{inequl}
\begin{aligned}
\frac{\partial u^{\ell}(\tilde \tau,  \mathbf{x}, \mathbf{y})}{\partial \tilde \tau} \geq \mathcal{A} u^{\ell}(\tilde \tau,  \mathbf{x}, \mathbf{y}), 
\end{aligned}
\end{equation}
where $\mathcal{A}(\cdot)$ follows \eqref{aineq}.
\end{thm}
The proof of this theorem can be found in Appendix~\ref{racp}. 

\subsection{Existence of Relative Arbitrage}
\label{ficherara}
The Cauchy problem \eqref{eqsol}-\eqref{eqsol2} admits a trivial solution $u^{\ell}(\tau,  \mathbf{x},  \mathbf{y}) \equiv e^{c_{\ell}}$. 
In this section, we discuss the conditions under which relative arbitrage exists. Specifically, one requires $u^{\ell}(\tau, \mathbf{x},  \mathbf{y}) < e^{c_{\ell}}$, which corresponds to the failure of uniqueness for the Cauchy problem.

For an admissible market model $\mathcal{M}$, we introduce the normalized reciprocal of the deflated total market capitalization \eqref{bigpi} and solve it using $L(t)$ in \eqref{ltmg}. Starting with $\mathcal{V}(0)$, it holds
\begin{equation}
\label{calvlhatt}
\begin{aligned}
   d \left( \mathcal{V}(t) L(t) \right) =  \mathcal{V}(t) L(t) \big(\Pi^{\prime}(t) \sigma(t) -  \theta'(t) \big)dW(t),
\end{aligned}
\end{equation}
where $\Pi(\cdot)$ is the corresponding portfolio of wealth $\mathcal{V}(t)$. We define $\Tilde{\theta}(\cdot) :=  \theta(\cdot) - \sigma'(\cdot) \Pi(\cdot)$, and $\widetilde{W}(\cdot) := W(\cdot) + \int_0^{\cdot} \Tilde{\theta}(t) dt$, then
\[
\Lambda(t) := \mathcal{V}(0)/L(t)\mathcal{V}(t) = \exp\bigg\{\int_0^t \Tilde{\theta}'(s) d\widetilde{W}(s) - \frac{1}{2} \int_0^t |\Tilde{\theta}(s)|^2 ds \bigg \}, \quad t \in [0,T].
\] 

We also introduce the stopping time of $\Lambda(t)$ touching zero
\[
\mathcal{T} := \inf \{t \geq 0 | \Lambda(t) = 0\} = \inf \{t \geq 0 | L(t)\mathcal{V}(t) = \infty\}.
\]

If $L(t)\mathcal{V}(t)$ is not a true martingale, then $L(T)\mathcal{V}(T)$ is not a Radon-Nikodym derivative, and the process $W(\cdot)$ is not necessarily a Brownian motion under an equivalent local martingale measure. Following the route suggested by \cite{opta} and \cite{ruf2}, there exists a probability measure $\mathbb{Q}$ on $(\Omega, \mathcal{F})$, such that $\mathbb{P}$ is locally absolutely continuous with respect to $\mathbb{Q}$: $\mathbb{P} <\!\!\!< \mathbb{Q}$, $\Lambda(T)$ is a $\mathbb{Q}$-martingale, and $d\mathbb{P} = \Lambda(T) d\mathbb{Q}$ holds on each $\mathcal{F}_T$, $T \in (0, \infty)$. Through the F\"ollmer exit measure \cite{exitm}, one can characterize the solution $u^{\ell}(\cdot)$ of the Cauchy problem in terms of the maximal $\mathbb{Q}$-probability that the supermartingale $L(T)\mathcal{V}(T)$ staying in $(0, \infty)^{2n}$ on $[0,T]$.
\begin{equation}
\label{mpc}
u^{\ell}(T) = e^{c_{\ell}}\mathbb{E}^{\mathbb{P}}[L(T)\mathcal{V}(T)]/\mathcal{V}(0) = e^{c_{\ell}} \mathbb{E}^{\mathbb{P}} \Big [ \textbf{1}_{\{\mathcal{T}>T \}} \frac{1}{\Lambda(T)} \Big] = e^{c_{\ell}} \mathbb{Q}(\mathcal{T} > T), \ \forall T \in [0,\infty).
\end{equation}

\begin{definition}[Auxiliary process and the Fichera drift]
\label{aux}
We define the following 
\begin{enumerate}
\item The auxiliary process $\zeta = (\zeta_1, \ldots, \zeta_{2n})$ is defined as 
\[
d \zeta_i(\cdot) = \widehat{b}_i (\zeta(\cdot)) dt + \sum_{k=1}^n \widehat{\sigma}_{ik} (\zeta(\cdot)) dW_k,\quad \zeta(0)=\zeta_0, \quad i =1, \ldots, 2n,
\]
where 
\[
\widehat{b}_i ( \mathbf{x}, \mathbf{y}) =
\begin{cases}
      \sum_{j=1}^n \frac{\delta  a_{ij} ( \mathbf{x}, \mathbf{y}) + (1-\delta) (s \tau')_{ij}( \mathbf{x}, \mathbf{y})}{\delta \mathbf{x} \cdot \mathbf{1} + (1-\delta) \mathbf{y} \cdot\mathbf{1}}  & \text{if i = 1, \ldots, n,}\\
      \sum_{j=1}^n  \frac{ (1-\delta) \psi_{ij} ( \mathbf{x}, \mathbf{y}) +  \delta (s \tau')_{ij}( \mathbf{x}, \mathbf{y}) }{\delta \mathbf{x} \cdot \mathbf{1} + (1-\delta) \mathbf{y} \cdot\mathbf{1}} & \text{if i = n+1, \ldots, 2n,}\\
    \end{cases}     
\]
\[\widehat{a}( \mathbf{x}, \mathbf{y}) = \widehat{\sigma}( \mathbf{x}, \mathbf{y}) \widehat{\sigma}'( \mathbf{x}, \mathbf{y}),\]
\begin{equation} 
\label{eq: sigmacombine}
\widehat{\sigma}_{ik}( \mathbf{x}, \mathbf{y}) = 
\begin{cases}
           s_{ik}( \mathbf{x}, \mathbf{y}) & \text{if} \ i,k = 1, \ldots, n,\\
       \tau_{ik}( \mathbf{x}, \mathbf{y}) & \text{if} \ i = n+1, \ldots, 2n, \ k = 1, \ldots, n,
    \end{cases}
\end{equation} 
for $( \mathbf{x}, \mathbf{y}) \in (0, \infty)^n \times (0, \infty)^n$.
 
\item The Fichera drift of $\zeta$ is defined as
\begin{equation}
\label{driftf}
    f_i(\mathbf{x}, \mathbf{y}):= \widehat{b}_i ( \mathbf{x}, \mathbf{y}) - \frac{1}{2} \sum_{j=1}^n D_j \widehat{a}_{ij} ( \mathbf{x}, \mathbf{y}),
\end{equation}
for $i = 1, \ldots, 2n$, $( \mathbf{x}, \mathbf{y}) \in (0, \infty)^n \times (0, \infty)^n$.
\end{enumerate}
\end{definition}

\begin{remark}
The auxiliary market portfolio ${\zeta_i}/{\sum_{j=1}^n \zeta_j}$ takes into account the interactions among investors. The induced measure $\mathbb Q$ corresponds to a change of drift under which the auxiliary market portfolio cannot be outperformed, since it satisfies the \textit{num\'eraire property} in the auxiliary system.

The maximal probability term is also related to the containment probability in the stochastic control literature. In \cite{pc}, an optimization problem is introduced to maximize the probability that the state trajectories remain in a bounded region over a given finite time horizon. This idea and stochastic control techniques have been used in \cite{uncertain} to consider the uncertainty of the optimal arbitrage problem (for one investor).
\end{remark}

Define $\mathcal{O}^{2n}$ as the boundary of $[0,\infty)^{2n}$. We construct the auxiliary process to be a $\left( [0,\infty)^{2n} \setminus \{\mathbf{0}\} \right)$-valued process. In the following, we show the condition for this process not to hit the limit of $[0,\infty)^{2n}$ during $[0,T]$.

\begin{asmp}
\label{uniquedist}
The system of $\zeta(\cdot)$ admits a weak, unique in distribution solution with values in $[0,\infty)^n \times [0,\infty)^n / \{\textbf{0}\}$,
{with given initial condition $\zeta(0) := (\mathbf{x}, \mathbf{y})$.}
\end{asmp}

The sufficient conditions for the above assumption are that $\widehat{b}(\cdot)$ is locally Lipschitz and Assumption~\ref{asmp1:main}, where the Lipschitz continuity can be relaxed to the local Lipschitz continuity \eqref{eq:locallip_tau} in Appendix \ref{apdx:fp}, and we can show that the covariance term in $\zeta(\cdot)$ also satisfies the growth condition.

\begin{prop}
\label{uniqpf}
Under Assumption~\ref{asmp1:main}-\ref{uniquedist}, suppose that the functions $\widehat{\sigma}_{ik}(\cdot)$ are continuously differentiable in $(0,\infty)^{2n}$ and that the matrix $\widehat{a}(\cdot)$ degenerates in $\mathcal{O}^{2n}$. {Assume that the degeneracy of $\widehat{a}(\cdot)$ and the Fichera drift of the process $\zeta(\cdot)$ can be extended by continuity in  $\mathcal{O}^{2n}_{\varepsilon^\star} := \bigcup_{i=1}^{2n} \{ z \in \mathbb{R}^{2n} : z_i \in [-\varepsilon^{\star}, 0]\}$, {for a small constant $\varepsilon^\star>0$}}.
For an investor $\ell$, if $f_i(\cdot) > 0$ holds on each face $i = 1, \ldots, 2n$ of the orthant, then $u^{\ell}(\cdot, \cdot) \equiv e^{c_{\ell}}$, and there is no arbitrage with respect to the benchmark portfolio on any time horizon.
If $f_i(\cdot) < 0$ on each face $\{x_i = 0\}$, $i = 1, \ldots, n$ and $\{y_i = 0\}$, $i = n+1, \ldots, 2n$ of the orthant, then $u^{\ell}(\cdot, \cdot) < e^{c_{\ell}}$ and arbitrage with respect to the benchmark portfolio exists, on every time horizon $[0, T]$ with $T \in (0,\infty)$.
\end{prop}

\begin{proof}
Take $\zeta_0$ in Definition~\ref{aux} as a given tuple $z := (\mathbf x, \mathbf y)$. As $\zeta(\cdot)$ is Markovian, we use a similar approach as in Theorem 2 of \cite{opta} to connect $u^{\ell}(T,  z)$ to the probability of the first hitting time of an auxiliary process to touch $\mathcal{O}^{2n}$. We denote $\mathcal{T}:= \inf \{t \geq 0 | \zeta(t) \in \mathcal{O}^{2n}\}$ as the first hitting time of the auxiliary process $\zeta(\cdot)$ to $\mathcal{O}^{2n}$. Using the nondegeneracy condition of $a_{ij}$, and \eqref{mpc}, 
we get that starting from $\zeta_0 = z$,
{
\[
u^{\ell}(T,  z) =  e^{c_{\ell}} \mathbb{Q}_{z}[\mathcal{T} > T], \quad 
(T, z) \in (0,\infty) \times (0,\infty)^{2n}.
\] 
}

{For the first claim, if the Fichera drift \eqref{driftf} satisfies $f_i(\mathbf{x}, \mathbf{y}) > 0$, for $i = 1, \ldots, n$, we only need to show that $\mathbb{Q}_z[\mathcal{T} > T] \equiv 1$, for $(T, z) \in (0,\infty) \times (0,\infty)^{2n}$. Note that $f_i(\cdot, \cdot)$ is continuous on $\mathcal{O}_{\varepsilon^\star}^{2n}$ for a small positive $\varepsilon^\star$. We take a positive constant $\varepsilon <\varepsilon^\star$ for the initial value $z$ and set $G_{R,\varepsilon}:=\{z \in \mathbb{R}^{2n}: z_i \leq -\varepsilon,\; |z|<R-\varepsilon, \text{for some \,} i = 1, \ldots, 2n\}$, where $R > |z| + \varepsilon$. Every point of $\partial G_{R,\varepsilon}$ lies strictly in the negative orthant of the coordinate plane. }

{We now construct a subset of $G_{R,\varepsilon}$ that has a connected $C^\infty$ boundary below.
Let $\rho(x)$ be the signed distance to $\partial G_{R,\varepsilon}$, that is, $\rho(x) = -\mathrm{dist}(x, \partial G_{R,\varepsilon})$ for $x \in G_{R,\varepsilon}$, and $\rho(x) = \mathrm{dist}(x, \partial G_{R,\varepsilon})$ if $x \notin G_{R,\varepsilon}$. 
Choose a standard mollifier $\eta$ as defined in \cite[Appendix C.5]{evans2010partial}. That is, $\eta(x) = C \exp \left( \frac{1}{|x|^2-1} \right)$, if $|x| < 1$; $\eta(x) = 0$ if $|x| \geq 1$. The positive constant $C$ is chosen such that $\int_{\mathbb{R}^{2n}} \eta(x) dx = 1$. 
Set $\eta_\delta(x) := \frac{1}{\delta^{2n}} \eta(\frac{x}{\delta})$. 
Define the mollification $\rho_\delta:=\rho*\eta_{\delta} \in C^{\infty}(\mathbb{R}^{2n})$. 
It is known that $\rho(\cdot)$ is 1-Lipschitz continuous, and  so \[
|\rho_\delta (x) - \rho(x)| \leq \int_{\mathbb{R}^{2n}} |\rho(x-y) - \rho(x)| \eta_\delta(y) dy \leq \int_{\mathbb{R}^{2n}} |y| \eta_\delta(y) dy \leq \delta.\]
By \cite{sard1942measure} (see also  \cite[Chapter 6]{lee2003smooth}), there exists a regular value $c (< -\delta)$ of the mollification  function $\rho_\delta$, and then we define the subset
\[\widetilde G_{R, \varepsilon}:=\{x\in \mathbb{R}^{2n}: \rho_\delta(x)<c\} \subsetneq \{x\in \mathbb{R}^{2n}: \rho(x)<0\} = G_{R, \varepsilon}.\]
with $|\nabla \rho_\delta(x)| > 0$  for $x \in \partial \widetilde G_{R, \varepsilon} =\{x\in \mathbb{R}^{2n}:\rho_\delta(x)=c\}$. 
Hence, according to the inverse function theorem \cite[Appendix C.6]{evans2010partial}, with a sufficiently large $R$, $\partial \widetilde G_{R, \varepsilon}$ is a connected $C^\infty$ boundary. }

{With $-2\delta \leq c \leq -\delta$, $|\rho_\delta - \rho| \leq \delta$, we are interested in the nonempty $\partial \widetilde G_{R, \varepsilon} \cap \mathcal O_{ \varepsilon^\star}^{2n}$ for some $\varepsilon^\star$ to be specified. As the Fichera drift $f_i(\cdot) > M > 0$ on each face of the orthant and $f_i(\cdot)$ is continuous on the compact set, 
we take a sufficiently small $\varepsilon^\star (>0) $ such that $f_i(\cdot) \geq 0$ on $\mathcal O_{ \varepsilon^\star}^{2n}$. Note that we need $[-\varepsilon - 3\delta, -\varepsilon + 3\delta] \subset [-\varepsilon^\star, 0]$ so that $\partial \widetilde G_{R, \varepsilon} \cap \mathcal O_{ \varepsilon^\star}^{2n} \neq \emptyset$. 
Because of $\varepsilon < \varepsilon^{\star}$, it holds
\[
\sum_{i=1}^n \bigg( \widehat{b}_i ( z) - \frac{1}{2} \sum_{j=1}^n D_j \widehat{a}_{ij} ( z)\bigg) \mathbf{n}_i < 0,
\]
in which $\mathbf{n}(z)$ 
is the outward normal vector at $z \in \mathcal{O}^{2n}_{\varepsilon^\star} \cap \partial \widetilde G_{R, \varepsilon}$, pointing to the exterior of $\widetilde G_{R, \varepsilon}$. 
{The process $\zeta(\cdot)$ starts from $(0, \infty)^{2n}$ and the first entrance to $\widetilde G_{R, \varepsilon}$ is through $\mathcal{O}^{2n}_{\varepsilon^\star} \cap \partial \widetilde G_{R, \varepsilon}$.}
Define the stopping time \[\tau_{R, \varepsilon} := \inf \{t \geq 0| \zeta(t) \in \widetilde G_{R, \varepsilon}\}.\]
From Theorem 9.4.1 (or Corollary 9.4.2) of \cite{fichera}, $\widetilde G$ is not attainable from the outside, $ B_R(0)/\widetilde G_{R, \varepsilon}$, that is, $\mathbb{Q}_z(\tau_{R, \varepsilon} > T) = 1$. Take a decreasing sequence $\{\varepsilon_j\}_{j = 1, 2, \ldots}$ such that $\varepsilon_j >0$, $\lim_{n \to \infty} \varepsilon_j = 0$. We have $\widetilde G_{R, \varepsilon_{j-1}} \subset \widetilde G_{R, \varepsilon_{j}}$, $\tau_R := \lim_{n \to \infty} \tau_{R, \varepsilon_j}$ is a stopping time and satisfies $\mathbb{Q}_z(\tau_{R} > T) = 1$. In other words, $\mathcal{O}^{2n} \cap \{|z| \leq R\}$ is not attainable by $\zeta(\cdot)$. Finally, letting $R \to \infty$, we get $\mathbb{Q}_z[\mathcal{T} > T] \equiv 1$, for $(T, z) \in (0,\infty) \times (0,\infty)^{2n}$.  
}

For the second claim, when $f_i(\cdot) < 0$ on each face $\{z_i = 0\}$, $i = 1, \ldots, 2n$, 
it is equivalent to showing that $\mathbb{Q}_z[\mathcal{T} > T] < 1$, for $(T, z) \in [0,\infty) \times [0,\infty)^{2n}$. We only need to show $\mathbb{Q}_z[\mathcal{T} < T] > 0$, i.e., the boundary $\{z_i = 0\}$, $i = 1, \ldots, 2n$, is attainable by $\zeta(\cdot)$. {For some constant $R > |z|$, 
$G_R := \bigcup_{i=1}^{2n} \{z \in \mathbb{R}^{2n}, z_i<0, |z|<R+\epsilon\}$. 
Consider that a subset ${\tilde{G}_R} \subsetneq G_R$ has a connected $C^3$ boundary $\partial \tilde{G}_R$ that touches $\mathcal O^{2n}$, i.e., $\partial \tilde{G}_R \cap \mathcal O^{2n} \neq \emptyset$. Pick $\epsilon > 0$ such that $\mathcal G := \overline{B_R(0) \setminus \tilde{G}_R}$ has a connected $C^3$ boundary. Denote $\mathcal{G}^o$ as the interior of ${\mathcal{G}}$. Note that $R$ is arbitrarily large; we focus on the part of the boundary close to $\mathcal{O}^{2n}$.}

By \cite[Chapter 11 Problems 7-8, Chapter 13 Definition 1]{fichera}, every point $z_0 \in \partial \mathcal{G}$ is a regular point, which means that, for every fixed $\delta > 0$, every $z_0 \in \partial \mathcal{G}$,
\begin{equation}
\label{regular}
    \lim_{\mathz \rightarrow z_0, \mathz \in \mathcal{G}^o} \mathbb{Q}_{\mathz} (\tau^g < \infty, |\zeta(\tau^g) - z_0|<\delta)=1,
\end{equation}
where $\tau^g$ is the exit time from ${\mathcal{G}}$. Define $\Sigma := \bigcup_{i=1}^{2n} \{z \in \mathbb{R}^{2n}: z_i = 0\} \bigcap {\partial \mathcal{G}}$. Thus, the Fichera drift is negative and the diffusion coefficient degenerates in $\Sigma$. The other part of the boundary is $\Sigma_2 := \Sigma^c \cap \partial \mathcal{G}$, where the diffusion coefficient is non-degenerate. {In the rest of the proof, we show that $\Sigma$, and thus $\bigcup_{i=1}^{2n}\{z_i = 0\}$,  is attainable by $\zeta(\cdot)$.}

Consider $z_0 \in \Sigma$, choose a small enough $\delta$ in \eqref{regular} such that $B_{\delta}^+(z_0) := \bigcap_{i=1}^{2n} \{z \in \mathbb{R}^{2n}: z_i > 0\} \bigcap B_{\delta}(z_0) \subset \mathcal{G}$. Then by \eqref{regular}, there exists $\eta > 0$ such that
whenever $|x-z_0| \leq \eta$, for an interior point $x \in \mathcal{G}^o,$ we have
\[
\mathbb{Q}_{x}(\tau^g < \infty, | \zeta(\tau^g) - z_0| < \delta)> \frac{1}{2}.
\]
As $B_\delta(z_0) \cap \partial \mathcal G \subset \Sigma$, we consider two cases: $ \lvert z - z_0 \rvert \le \eta$ and $\lvert z - z_0\rvert > \eta$: 
\begin{itemize}
    \item For $|z - z_0| \leq \eta$, the probability of reaching the boundary $\Sigma$ is positive, i.e.,
    \[
    \mathbb{Q}_{z}(\tau^g < \infty, \ \zeta(\tau^g) \in \Sigma)>0.
    \]
    
\item For $|z - z_0| > \eta$, 
\[
\inf_{\mathz \in A} \mathbb{Q}_{\mathz}(\tau^g<\infty, |\zeta(\tau^g) - z_0| < \delta) > \frac{1}{2},
\]
where
\[
A:= \bigcap_{i=1}^{2n} \{z \in \mathbb{R}^{2n}: z_i > 0, |z-z_0|=\eta\}.
\]
We have $\{z\} \bigcap A = \emptyset$, and we choose a subset $\mathcal{H} \subset \mathcal{G}$ that contains $\{z\} \bigcup A$. 
Define $\tau_{\star} := \inf\{t>0: \zeta(t) \notin \mathcal{H} \setminus A \}$, that is, $\tau_{\star}$ is the exit time of $\zeta(\cdot)$ from the bounded domain $\mathcal{H} \setminus A$. $\tau_{\star} < \infty$ almost surely since the diffusion coefficients are non-degenerate and uniformly elliptic in the interior of $\mathcal{G}$, and the drift coefficients are finite.
Take a continuous deterministic path $\omega_{\star}$ such that an $\epsilon$-neighborhood $N_{\epsilon, \omega_{\star}}$ of $\omega_{\star} \in C([0,T]; \mathbb{R}^{2n})$ for some $\epsilon > 0$ is contained in $ \{\omega \in \Omega: \zeta(\tau_{\star}(\omega), \omega) \in A\}$ with $\omega(0)=z$, $\omega(\tau_{\star})=r \in A$, and $\omega(s) \notin A$ for $0 \leq s < \tau_{\star}$.
That is, 
\[
N_{\epsilon, \omega_{\star}} = \{\omega \in C([0,T]; \mathbb{R}^{2n}): \ \omega(0) = z, \ \|\omega - \omega_{\star}\|<\epsilon\} \subset \{\omega \in \Omega: \zeta(\tau_{\star}(\omega), \omega) \in A\},
\]
where $\|\cdot\|$ is the supremum norm $\|\omega_1 - \omega_2\| = \sup_{0 \leq s \leq \tau_{\star}} |\omega_1(s) - \omega_2(s)|$, $\omega_1$, $\omega_2 \in C([0,T]; \mathbb{R}^{2n})$.
By the support theorem \cite[Lemma 3.1]{supthm} (see also \cite[Exercise 6.7.5]{supthm2}), 
\[
\mathbb{Q}_{z}(N_{\epsilon, \omega_{\star}})>0, \quad \text{and hence\ }\mathbb Q_z \left( \zeta(\tau_{\star}) \in A, \tau_{\star} < \infty \right) >0.
\]

Therefore, we have 
\[
\begin{aligned}
\mathbb{Q}_{z}(\zeta(\tau^g) \in \Sigma, \tau^g < \infty) & \geq \mathbb{Q}_{z} (|\zeta(\tau^g) - z_0| < \delta, \tau^g < \infty)\\
& \geq \mathbb{E}_{z} [\mathbf{1}(\zeta(\tau_{\star}) \in A, \tau_{\star}<\infty) \mathbb{E}_{z}[\mathbf{1}(|\zeta(\tau^g) - z_0| < \delta, \tau^g < \infty)| \mathcal{F}_{\tau_{\star}})]]\\
& = \mathbb{E}_{z} [ \textbf{1}(\zeta(\tau_{\star}) \in A,\ \tau_{\star}<\infty) \mathbb{Q}_{\zeta(\tau_{\star})} (|\zeta(\tau^g) - z_0| < \delta, \tau^g < \infty) ]\\
& \geq \mathbb{E}_{z} [  \textbf{1}(\zeta(\tau_{\star}) \in A, \ \tau_{\star}<\infty) \inf_{\mathz \in A}\mathbb{Q}_{\mathz} (|\zeta(\tau^g) - z_0| < \delta, \tau^g < \infty)]\\
& > \frac{1}{2} \mathbb{Q}_{z} (\zeta(\tau_{\star}) \in A, \ \tau_{\star}<\infty) > 0.
\end{aligned}
\]
The equality in the above expressions comes from the strong Markov property of $\zeta(\cdot)$.
\end{itemize}

Since $R$ is arbitrary, the process $\zeta(\cdot)$ reaches the set $\bigcup_{i=1}^{2n} \{z_i = 0\}$ with positive probability. In conclusion, $u^{\ell}(\cdot, \cdot)<e^{c_{\ell}}$ when Fichera drifts are negative, i.e., $f_i(\cdot) < 0$ on each face of $\mathcal{O}^{2n}$.
\end{proof}

As a corollary, if $f_i(\cdot) < 0$ on each face of $\mathcal{O}^{2n}$, investor $\ell$ can exploit relative arbitrage opportunities with a unique $u^{\ell}(\cdot)$, the minimal solution of \eqref{inequl}. 
The case of zero Fichera drift $f_i(\cdot) = 0$ on each face of the orthant is not treated in Proposition \ref{uniqpf} because of the non-smoothness of the boundary. It is conjectured that a similar reasoning through the Lyapunov-type function arguments used in \cite[Chapter 11, Section 6]{fichera} and \cite[Chapter 6, Section 1]{pinsky1995positive} for non-attainability may provide a proof for such a case. {In our proof of Proposition \ref{uniqpf}, we construct a connected $C^\infty$ boundary considering a relaxed region $\mathcal{O}^{2n}_{\varepsilon^\star}$ with a sufficiently small $\varepsilon^\star >0 $ so that the Fichera drift $f_i(\cdot) \geq 0$ on $\mathcal O_{ \varepsilon^\star}^{2n}$.}

In the next section, we derive the corresponding optimal strategy process $\pi^{\ell}$ for each $\ell = 1, \ldots, N$.

\section{The N player game}
\label{sec: nplayergame}
\numberwithin{equation}{section}

{In the previous sections, we focused on the relative arbitrage opportunity of an investor. We now turn to the second question posed in the Introduction (Section~\ref{intro}): If there exists a strategy to achieve optimal arbitrage opportunities, is it possible for all $N$ investors to follow it?}

{Stock capitalizations and investors' wealth processes are coupled, since the strategies adopted by the group of competitive investors contribute to the change of the trading volume of each stock, and thus to the change of stock capitalization. As we show below, this coupling leads naturally to fixed-point problems on the path space of strategies and motivates the formulation of an $N$-player game.}

We model the investors as participants in an $N$-player game and search for optimal strategies in the $N$-player game. We are interested in the existence and uniqueness of Nash equilibria, since this informs us of the incentives that investors face to change their strategies, as well as the challenges as a rational investor to predict the market. If Nash equilibrium does not exist, investors cannot find such stable relative arbitrage opportunities in a fixed amount of time. {If multiple Nash equilibria exist, then both decision-making and market prediction become more complicated. In that case, one may need to consider alternative mechanisms, such as cooperation among players, in order to recover a meaningful notion of optimality.} This will be discussed in Section~\ref{sec: future}.

The following sections are discussed under Assumptions~\ref{asmp1:main}-\ref{uniquedist}.

\subsection{Construction of Nash equilibrium} 
\label{uniquenessne}
The solution concept of this $N$ player game is the Nash equilibrium. {In particular, given the strategies chosen by the other players,} a typical player seeks the best response to all the other players, which amounts to the solution of an optimal control problem to minimize the expected cost $u^{\ell}$. Specifically, when an investor $\ell$ assumes that the wealth of other players is fixed, they wish to take the solution of \eqref{eqsol} and \eqref{eqsol2} as their initial wealth so that
\begin{equation}
\label{muexplain}
    V^{\ell}(T) \geq e^{c_{\ell}} \mathcal{V}(T) = \delta \cdot e^{c_{\ell}} X(T) + (1-\delta) \cdot e^{c_{\ell}} \frac{1}{N}\sum_{\ell=1}^N V^{\ell}(T).
\end{equation}
Thus, we define the cost functional
\[
J^\ell (\text{\boldmath$\pi$}) := \inf \bigg \{ \omega^{\ell} > 0 \  \big| \, V^{\omega^{\ell} \mathcal{V}(0), \pi^{\ell}}(T) \geq e^{c_{\ell}} \mathcal{V}(T)\bigg \},
\]
for all admissible strategy profiles $\text{\boldmath$\pi$}(\cdot) = (\pi^1(\cdot), \ldots, \pi^N(\cdot))$, $\text{\boldmath$\pi$} \in \mathbb{A}^{(N)}$. The influence of $\pi^{k}$, $k \neq \ell$, is implicitly defined in the wealth $V(\cdot)$ and the benchmark $\mathcal{V}(\cdot)$.

We now define Nash Equilibrium over the entire time horizon $[0,T]$.
\begin{definition}[Nash Equilibrium]
\label{nee}
For each investor $\ell=1, \ldots , N $, take $\pi^{\ell \star} = (\pi^{\ell \star}_1, \ldots, \pi^{\ell \star}_n) \in \mathbb A$ as the admissible strategies in Definition~\ref{portfoliopi}. The strategy profile $\text{\boldmath$\pi$}^{\star}(\cdot) = (\pi^{1 \star}(\cdot), \ldots, \pi^{N \star}(\cdot))$ is called a Nash equilibrium over $[0,T]$, if, for every $\pi^{\ell} \in \mathbb{A}$,
\begin{equation} \label{eq: NE def}
J^\ell(\pi^{\ell \star}, \pi^{-\ell \star}) \leq J^\ell(\pi^{\ell}, \pi^{-\ell \star}),\quad \ell = 1, \ldots , N \, , 
\end{equation}
and $\pi^{-\ell}$ is the subset of the strategy profile without investor $\ell$, $\pi^{-\ell}(\cdot) = (\pi^1(\cdot), \ldots, \pi^{\ell-1}(\cdot), \pi^{\ell+1}(\cdot), \ldots, \pi^N(\cdot))$. For $t\ge 0 $, we define the empirical measures of the corresponding wealth $\big( V^{\ell\star}(t)\big)_{\ell = 1, \ldots, N} \in \mathbb{R}_+^N $ of the investor $\ell$ with the Nash equilibrium strategy $\pi^{\ell\star}$, given the initial measure $\mu_0 \in \mathcal{P}_2(\mathbb{R}_+)$ by 
\begin{equation} \label{eq: NE def2}
\mu_t^\star := \frac{1}{N} \sum_{\ell=1}^N {\bm \delta}_{V^{\ell\star}(t)},
\end{equation}
where ${\bm \delta}_x$ is the Dirac delta mass at $x \in \mathbb R_+$. We write the Nash equilibrium as $(\pi^\star, \mu^\star)$, where $\text{\boldmath$\pi$}^{\star} \in C([0,T];\Delta^N)$, $\mu^{\star} \in \mathcal{P}_2(C([0,T]; \mathbb{R}_+ ))$.
\end{definition}

Each individual aims to minimize the relative amount of initial capital so that one can match or exceed the benchmark at the terminal time. From the previous section it follows that $\inf_{\pi^{\ell} \in \mathbb{A}} J^{\ell}(\pi) = u^{\ell \star}(T)$, where $u^{\ell \star}$ is the minimum nonnegative solution of \eqref{inequl}.

Investors focus more on changes in the wealth processes of other investors than changes in their strategies, because two distinct strategy processes can lead to identical wealth at the same time $T$.
{For this reason, uniqueness in distribution at the level of wealth is the more relevant notion. We therefore adopt the following notion of uniqueness of Nash equilibrium, formulated in terms of wealth processes.}
\begin{definition}[Uniqueness]
\label{def: uni-ne}
Consider two sets of Nash equilibrium strategies $\pi_a$ and $\pi_b$ such that \eqref{eq: NE def} holds. The corresponding empirical measure flows $\mu_a := (\mu^a_t)_{t \in [0,T]}$, $\mu_b:= (\mu^b_t)_{t \in [0,T]}$ are defined on the filtered probability space $(\Omega, \mathcal{F}, \mathbb{F}, \mathbb{P})$, with the same initial law $\mu_0 \in \mathcal{P}_2(\mathbb{R}_+)$. We say that the Nash equilibrium is unique if the measure flows $\mu_a$ and $\mu_b$ are indistinguishable, that is,
\[
\mathbb{P}[\mu_a = \mu_b] = 1.
\]
\end{definition}

In general, the market $\mathcal{X} \in C([0,T]; \mathbb{R}_+^n)$ and investors' wealth $\mathbf{V} \in C([0,T]; \mathbb{R}_+^N)$ interact through two mean-field interaction terms, the joint empirical measure $\nu$ of wealth and strategies in Definition~\ref{nunt} and the empirical measure of wealth. In the special case \eqref{muexplain}, the latter becomes the empirical mean of wealth. Denote the interactions $\mu$ and $\nu$ in the Nash equilibrium of games of $N$ players by $\mu^{\star}$ and $\nu^{\star}$. Generally, the uniqueness of $\mu^{\star}$ is a weaker condition than the uniqueness of $\nu^{\star}$. If there is a unique optimal $\nu^{\star}$ in the sense of Definition~\ref{nee}, then it implies that its marginal distribution $\mu^{\star}$ is the unique Nash equilibrium defined by Definition~\ref{def: uni-ne}. However, the converse is not true: A unique $\mu^{\star}$ does not necessarily give a unique optimal $\nu^{\star}$. For each $\ell$, solving the corresponding $\{\pi_i^{\ell}(t)\}_{i = 1, \ldots, n}$ of the unique $V^{\ell \star}(t)$ relies on the Malliavin calculus and the solution can be written by different stochastic processes. Thus, there could be multiple possible quantities of the optimal measure $\nu^{\star}$, and multiple solutions of the strategy profile $\pi^{\ell \star}(\cdot)$ that generated the unique $V^{\ell \star}$ or $\mu^{\star}$ for $\ell = 1, \ldots, N$. See the general setup of the joint empirical measure $\nu$ and the general particle system in Appendix~\ref{nmckean}.

\subsection{Optimal arbitrage opportunities and the corresponding strategies in \texorpdfstring{$N$}{N}-player game}
\label{42}

All $N$ players have the same goal of competing with the market and other participants to pursue relative arbitrage opportunities. Essentially, the players pursue the optimal strategy in the Nash equilibrium in order to reach the optimal initial amount of investment for each player. If the Nash equilibrium solution satisfies the condition about the Fichera drift in Proposition~\ref{uniqpf}, then there is a relative arbitrage opportunity for each investor. 

\subsubsection{Searching Nash equilibrium} 
\label{sec4.2.1}

{We next derive the Nash equilibrium and the associated optimal strategies.} Here, we assume the Markovian structure of the admissible strategies, that is, the strategy that an investor $\ell$ adopts at time $t$ is in the form of {$\pi^{\ell \star}_i(t) = \phi^{\ell}(\mathbf{x},\mathbf y)\vert_{(\mathbf x, \mathbf y) = (\mathcal X(t), \mathcal Y(t))}$.}

We specify the methodology to find the optimal path in the sense of the equilibrium of $N$ players in the following. 
Notice that we consider the $N$ player game in a dynamic programming fashion, where we solve the subproblems specified in Section~\ref{dp} for every $t \in [0,T)$. We solve the corresponding strategies and ultimately determine the optimal initial wealth to achieve relative arbitrage as defined in Definition~\ref{uTdef}.

Solving the Nash equilibrium takes the following steps:
\begin{enumerate}
    \item Suppose that we start with a given set of control processes $\text{\boldmath$\pi$}:=(\pi^1, \ldots, \pi^N) $, where $\pi^{\ell}(\cdot)$ is of the form $\phi^{\ell}(\mathcal{X}(\cdot),\mathcal{Y}(\cdot))$, $\ell = 1,\ldots, N$. 
    \begin{enumerate}
        \item Solve the $N$-particle system \eqref{eq: x}-\eqref{eq: y}, whose coefficients are determined by It\^o's formula of the function $\phi^{\ell}(\mathcal{X}(\cdot),\mathcal{Y}(\cdot))$, $\ell = 1,\ldots, N$. The detail is included in \eqref{gonc1}.
    \item Solve $u^{\ell}(T-t) := \inf_{\text{\boldmath$\pi$} \in \mathbb{A}^{(N)}} J^{\ell}(\text{\boldmath$\pi$})$ through the nonnegative minimal solution of the linear PDE similar to \eqref{eqsol}, with $u^{\ell}(0,\mathbf{x}, \mathbf{y}) = e^{c_{\ell}}$, $\ell = 1, \ldots, N$. {For each investor $\ell$, the strategies $\pi^{-\ell}$ are fixed.}
    \end{enumerate}
    \item With the solution $\{u^{\ell}(T-t)\}_{t \in [0,T], \ell = 1, \ldots, N}$, determine the corresponding optimal control $\hat{\text{\boldmath$\pi$}}$. Thus, we can find a map $\Phi$ such that $\hat{\text{\boldmath$\pi$}} = \Phi(\text{\boldmath$\pi$})$. We specify the existence of such $\Phi(\cdot)$ in Theorem~\ref{thm: neforxy}. Note that the fixed point mapping $\Phi(\cdot)$ is generally not unique, as explained in Remark~\ref{cpdfs}. 

    \item If there exists $\text{\boldmath$\pi$}^{\star}$ such that $\text{\boldmath$\pi$}^{\star} = \Phi(\text{\boldmath$\pi$}^{\star})$, then the pair $(\text{\boldmath$\pi$}^{\star}, \mu^{\star})$ is the Nash equilibrium, where  $v^{\ell \star}$ is  specified through \eqref{v0v0} and $\mu^{\star} := \frac{1}{N} \sum_{\ell=1}^N {\bm \delta}_{(V^{v^{\ell\star}, \pi^{\ell} \star} )}$.
\end{enumerate}

\subsubsection{Fixed point problem}
\label{sec4.2.2}
{We now derive an explicit representation of the optimal strategy in terms of market capitalization, trading volume, and the benchmark portfolio.}
\begin{thm}[Fixed point problem]
\label{thm: neforxy}
Fix $\delta \in (0,1]$. Under Assumption~\ref{asmp1:main}-\ref{uniquedist},
{at each time $t$, the $\ell$-th component of the best response map $\Phi(\text{\boldmath$\pi$})$ satisfies}
\begin{equation}
\label{eq:bestresponse}
\begin{aligned}
\Phi(\text{\boldmath$\pi$})_{t, \ell} & = \mathcal{X}(t) D_x \log u^{\ell, \pi}(T-t, \mathbf{x}, \mathbf{y})+ \tau( \mathbf{x}, \mathbf{y}) \sigma^{-1}(\mathbf{x}, \mathbf{y}) D_{y}  \log u^{\ell, \pi}(T-t, \mathbf{x}, \mathbf{y})\\
& \ \ +\frac{\delta X(t)}{\mathcal{V}(t)} \mathbf{m}(t) + \frac{(1-\delta)}{N \mathcal{V}(t)} \sum_{k = 1}^N V^{k, \pi}(t) \pi^{k}(t), \quad t \ge 0. 
\end{aligned}    
\end{equation}
For $i = 1, \ldots, n$, the optimal strategy follows
\begin{equation}
\label{anotherpi2}
\begin{aligned}
\pi^{\ell \star}_i(t) 
&= \left( X_i(t) D_{x_i} \log u^{\ell, \text{\boldmath$\pi$}^{\star}}(T-t, \mathbf{x}, \mathbf{y})+ \tau_i( \mathbf{x}, \mathbf{y}) \sigma^{-1}(\mathbf{x}, \mathbf{y}) D_{y_i}  \log u^{\ell, \text{\boldmath$\pi$}^{\star}}(T-t, \mathbf{x}, \mathbf{y}) + \Pi^{\star}_i(t) \right) \bigg \vert_{(\mathbf x, \mathbf y) = (\mathcal X(t), \mathcal Y(t))}.
\end{aligned}
\end{equation}
At terminal time $T$, $\pi_i^{\ell \star}(T) = \Pi^{\star}_i(T)$, for each $\ell = 1,\ldots, N$. 
We use the notation $u^{\ell, \pi}(\cdot)$ to emphasize that the coefficients of the Cauchy problem depend on $\pi$. Similarly, $V^{k, \pi}(t)$ is generated from $\pi^{k}(t)$, $k = 1, \ldots, N$. $\Pi^{\star}(t)$ is the benchmark portfolio that replicates $\mathcal{V}(t)$ in Proposition~\ref{pistarr} at the best response. When $\delta = 1$, that is, when we search for the best strategies only to outperform the market, $\pi_i^{\ell \star}(t) = \mathbf{m}^{\star}_i(t)$, for each $\ell = 1,\ldots, N$. 
\end{thm}



{
Note that in \eqref{anotherpi2}, $\tau(\cdot)$ is determined by It\^o's formula of the function $\phi(\mathbf{x},\mathbf{y})$, which yields another fixed point problem specified in Appendix~\ref{apdx:fp}. Because an explicit solution of the optimization problem is technically difficult, the PDE analysis in Section~\ref{5.1} is carried out under the time-homogeneous form of $(\tau, \gamma)(\cdot)$ and that the strategy profile of each investor depends only on $(\mathcal{X}, \mathcal{Y})$. This may hold in special cases, for instance, when the minimal solution of Cauchy PDE~\eqref{eqsol} is separable in time and space.
In the general time-inhomogeneous setting, the corresponding formulation is described in Appendix~\ref{nmckean}.}

\begin{proof}[Proof of Theorem~\ref{thm: neforxy}]
For every $\ell = 1, \ldots, N$, the player $\ell$ reacts to changes in the market (the strategies $\pi^{-\ell}$ are fixed) and adopts the best response $\Phi(\text{\boldmath$\pi$})$ that achieves 
\begin{equation}
\label{optvstarequ}
    V^{\ell \star}(\cdot) = \mathcal{V}(\cdot) u^{\ell}(T-\cdot).
\end{equation}  
The Markovian property \eqref{ggg} implies the deflated wealth process
\begin{equation}
    \label{mgVhat}
    \hat{V}^{\ell \star}(t) := V^{\ell \star}(t) L(t)
= \mathbb{E} \big[ \mathcal{V}(T) L(T)|\mathcal{F}_t \big]
\end{equation}
is a martingale. As a result, from \eqref{mgVhat}, the $dt$ terms in $d\hat{V}^{\ell}(t) = d(\mathcal V^{\ell}(t) L(t) u^{\ell}(T-t))$ must vanish; that is,
\begin{equation}
\label{nt2}
\hat{V}^{\ell \star}(t) = \hat{V}^{\ell \star}(0) + \sum_{k=1}^n \int_0^t \hat{V}^{\ell \star}(s) B_k(T-s, \mathcal{X}(s), \mathcal{Y}(s)) dW_k(s), \ 0 \leq t \leq T,    
\end{equation}
where for $\rho = T - t$, $t \in [0,T]$, 
\[
\begin{aligned}
B_k(\rho, \mathbf{x}, \mathbf{y}) 
&= \sum_{i=1}^n \sigma_{ik}(\mathbf{x}, \mathbf{y}) x_i D_i \log u^{\ell}(\rho, \mathbf{x}, \mathbf{y}) + \sum_{m=1}^n \tau_{mk}(\mathbf{x}, \mathbf{y}) D_{m} \log u^{\ell}(\rho, \mathbf{x}, \mathbf{y}) \\
&\ \ + \sum_{i=1}^n \frac{\delta X(t)}{\mathcal{V}(t)} \bigg(\frac{x_i}{\sum_{i=1}^n x_i}\sigma_{ik}(t) - \theta_k(\mathbf{x}, \mathbf{y}) \bigg) \\
& \ \ + \frac{(1-\delta)/N}{\mathcal{V}(t)} \sum_{i=1}^n \sum_{\ell = 1}^N \bigg(  V^{\ell \star}(t) \pi^{\ell}_i \sigma_{ik}(t) -  V^{\ell \star}(t)\theta_k(\mathbf{x}, \mathbf{y}) \bigg) \Big \vert_{(\mathbf x, \mathbf y) = (\mathcal X(t), \mathcal Y(t))}.
\end{aligned}
\]
Hence, the best response gives the strategy that replicates \eqref{nt2}. That is, assume that all controls $\pi^{k}(\cdot)$, $k \neq \ell$ are chosen, and player $\ell$ will choose the optimal strategy by comparing the general formula $\hat{V}^{\ell}$ in \eqref{vlhatt} and $\hat{V}^{\ell \star}$ in \eqref{nt2}. Thus, we derive that \eqref{eq:bestresponse} holds.


With a fixed set of control processes $\{\pi^{\ell}(t)\}_{0 \leq t \leq T}$, we solve $u_{T-t}^{\ell}$ for $t \in [0,T]$, and expect that the optimal strategy $\pi^{\ell \star}$ will coincide with the fixed $\pi^{\ell}(\cdot)$ for each $\ell$. Every player $k = 1, \ldots, N$ acts homogeneously, so that their wealth follows $V^{k \star}(\cdot)= \mathcal{V}(\cdot) u^{k}(T-\cdot)$. 
With an optimal strategy $\pi^{\ell \star}(\cdot)$, the wealth of player $\ell$ matches $e^{c_{\ell}} \mathcal{V}^N(\cdot)$. Hence, for a candidate portfolio vector $(\pi^1(\cdot), \ldots, \pi^N(\cdot))$ to be a Nash equilibrium, we need the $\pi^{\ell}(\cdot)$ to be identical to $\pi^{\ell \star}(\cdot)$, for $\ell = 1, \ldots, N$. That is, we solve for the fixed point problem for $\pi^{\ell \star}$ 
\begin{equation}
\label{anotherpi}
\begin{aligned}
\pi^{\ell \star}_i(t) &= X_i(t) D_i \log u^{\ell}(T-t, \mathbf{x}, \mathbf{y})+ \tau_i(\mathbf{x}, \mathbf{y}) \sigma^{-1}(\mathbf{x}, \mathbf{y}) D_{k}  \log u^{\ell}(T-t, \mathbf{x}, \mathbf{y})\\
& \ \ +\frac{\delta X(t)}{\mathcal{V}(t)} \mathbf{m}_i(t) + \frac{(1-\delta)}{N \mathcal{V}(t)} \sum_{\ell = 1}^N V^{\ell \star}(t) \pi^{\ell \star}_i(t),
\end{aligned}
\end{equation}
where $V^{\ell \star}(t)$ is generated from $\pi^{\ell \star}(t)$. Equation \eqref{anotherpi2} then follows from \eqref{anotherpi}, showing explicitly that the optimal strategy decomposes into a market component, a trading-volume component, and the benchmark portfolio.

\end{proof}

{To illustrate the fixed point problem in Theorem \ref{thm: neforxy}, we consider the following example. If the market coefficients are of the form
    \begin{equation}
    \label{eq: umn}
    (b_i, s_{ik}, \gamma_i, \tau_{ik})(\mathcal{X}(t), \mathcal{Y}(t)) = (\widetilde{b}_i, \widetilde{s}_{ik}, \widetilde{\gamma}_i, \widetilde{\tau}_{ik})(\mathbf{m}(t), \mathbf{n}(t)), \quad 1\le i,k \le n , \,\, t \ge 0 ,      
    \end{equation}
    where $\mathbf{m}(\cdot)$ is the market portfolio and $\mathbf{n}_i(\cdot) := \frac{\mathcal{Y}_i(\cdot)}{Y(\cdot)}$, $\mathbf{n}(\cdot) := (\mathbf{n}_1(\cdot), \ldots, \mathbf{n}_n(\cdot))$ for some functions $\widetilde{b}_\cdot$,$\widetilde{s}_\cdot $, $ \widetilde{\gamma}_\cdot$, $\widetilde{\tau}_\cdot$. $Y(\cdot) := \sum_{i=1}^n \mathcal{Y}_i(\cdot)$.
Then we can derive the following. 
\[
\begin{aligned}
    X_i(t) D_{x_i} \log u^{\star}(T-t, \mathbf{x}, \mathbf{z}) = \mathbf{m}_i \Big(D_{m_i} &\log \widetilde{U}(T-t, \mathbf{m}(t), \mathbf{n}(t)) \\
    &- \sum_{j=1}^n \mathbf{m}_j D_{m_j} \log \widetilde{U}(T-t, \mathbf{m}(t), \mathbf{n}(t))\Big),
\end{aligned}
\]
\[
\begin{aligned}
\mathcal Y_i(t) D_{y_i} \log u^{\star}(T-t, \mathbf{x}, \mathbf{z}) = \mathbf{n}_i \Big(D_{n_i} &\log \widetilde{U}(T-t,\mathbf{m}(t), \mathbf{n}(t)) \\
&- \sum_{j=1}^n \mathbf{n}_j D_{n_j} \log \widetilde{U}(T-t,\mathbf{m}(t), \mathbf{n}(t))\Big),
\end{aligned}
\]
in \eqref{anotherpi2}. Hence, assuming $\widetilde{U}(\cdot) \in C^{1,3,3}([0,T] \times \mathbb{R}^ n_+ \times \mathbb{R}^n_+)$, we find that the optimal strategy in \eqref{anotherpi2} satisfies $\sum_{i=1}^n \pi^{\star}_i(t) = 1$. 
Under Assumption~\ref{asmp1:main}-\ref{uniquedist}, $u^{\ell}(\tilde \tau, \mathbf{x}, \mathbf{y})\in C^{1,3,3}([0, T] \times (0, \infty)^n \times (0, \infty)^n)$ is bounded and $\Phi$ is a continuous mapping from $C([0,T]\times \mathbb R^n_+ \times \mathbb R^n_+; \Delta^N)$ to $C([0,T]\times \mathbb R^n_+ \times \mathbb R^n_+; \Delta^N)$. Consider an increasing sequence of compact subsets $[0,T]\times K^{(m)}_1 \times K^{(m)}_2$, for each compact set $K^{(m)}_1,K^{(m)}_2 \subset \mathbb{R}^n$; $K^{(m)}_j \subset K^{(m+1)}_j$ for $j = 1,2$, $\bigcup_{m=1}^\infty K_j^{(m)} = \mathbb{R}_+^n$. Assume $u^{\ell}(\cdot)$ and its first time derivative, first- and second order spatial derivatives are uniformly bounded in $\pi(\cdot)$, and assume $\tau(\cdot)$ and $\sigma(\cdot)$ are bounded. Let $\Phi^{(m)}(\text{\boldmath$\pi$})$ be the restriction of $\Phi(\text{\boldmath$\pi$})$ to $[0,T]\times K^{(m)}_1 \times K^{(m)}_2$. Then, $\Phi^{(m)}(\text{\boldmath$\pi$})$ in \eqref{eq:bestresponse} is bounded and equicontinuous on $[0,T]\times K^{(m)}_1 \times K^{(m)}_2$, $m \in \mathbb{N}$. The relatively compactness of $\Phi^{(m)}(\text{\boldmath$\pi$})$ in $C([0,T]\times K^{(m)}_1 \times K^{(m)}_2; \Delta^N)$ then follows by Arzel\`a-Ascoli Theorem. By Schauder fixed point theorem, the fixed point solution $\text{\boldmath$\pi$}^{(m)\star}$ exists such that $\text{\boldmath$\pi$}^{(m)\star} = \Phi^{(m)}(\text{\boldmath$\pi$}^{(m)\star})$. We can use Cantor's diagonal argument to find a subsequence $\text{\boldmath$\pi$}^{(m_k)\star} \to \text{\boldmath$\pi$}^{\star}$ as $m_k \to \infty$. Then $\text{\boldmath$\pi$}^{\star} = \Phi(\text{\boldmath$\pi$}^{\star})$ and $\left( \text{\boldmath$\pi$}^{\star}, \mu^{\star} := \frac{1}{N} \sum_{\ell=1}^N {V^{v^{\ell}, \text{\boldmath$\pi$}^{\ell \star}}} \right)$ is the Nash equilibrium.
$\{u^{\ell}(T, \mathbf{x}, \mathbf{y})\}_{\ell = 1, \ldots, N}$ is the value function of the $N$-player game under equilibrium. 
}

As shown in Definition~\ref{def: uni-ne} and the subsequent statements, there are many strategies that lead to the same optimal empirical wealth distribution $\mu^{\star}$. It is challenging and unnecessary to search for all the possible strategies. Theorem~\ref{thm: neforxy} provides an example of the optimal strategy. This structure of optimal strategy is similar to functional generated portfolios, as studied in \cite{thesis}, which yields practical implementations and data-driven methods for stochastic portfolio theory. To this end, we provide the following example.

\begin{example}
\label{ex: vsm}
We construct the stock capitalization coefficients using the similar idea in volatility-stabilized market models (\cite{ra2005}). The main characteristics of volatility-stabilized market models are the \textit{leverage effect}, where the rate of return and volatility have a negative correlation with the stock capitalization relative to the market $\{\mathbf{m}_i(t)\}_{i = 1, \ldots, n}$. 
Smaller stocks tend to have higher volatility than larger stocks. The coefficients $\beta(\cdot)$ and $\sigma(\cdot)$ in $\mathcal{M}$ are set to the following specific forms that agree with these market behaviors. For $1 \le i, j \le n$, with a given number of investors,
\begin{equation}
\label{eq: coeff-vsm}
    \beta_i(t) = \frac{C_x}{\mathbf{m}_i(t) \mathcal Y_i(t)} , \quad a_{ij} = \frac{X_i(t)}{\mathcal Y_i(t)} X(t)\delta_{ij},
\end{equation}
where $\delta_{ij} = 1$, when $i = j$; and $\delta_{ij} = 0$ otherwise, when $i \neq j$. $C_x$ is a given nonnegative constant. $\mathbf{m}(\cdot)$ is the market portfolio.

Let $\delta =\frac{1}{2}$, and consider a simplified market structure with 
\[d \mathcal Y_t = y_0 dt,\] 
where $\mathcal{V}_0$ is defined as $\mathcal V_0 = \frac{x_0}{2 - \frac{1}{N} \sum_{k=1}^N u^k(T)}$, and $y_0 = \frac{x_0}{2 - \frac{1}{N} \sum_{k=1}^N u^k(T)}- \frac{x_0}{2}$. Then the diffusion coefficient $\tau(\mathbf{x},\mathbf{y}) = 0$, for any $(\mathbf{x},\mathbf{y}) \in (0, \infty) \times (0, \infty)$. Hence, we adapt Assumption~\ref{hasmp} to the existence of a function $H: \mathbb{R}_+^n \rightarrow \mathbb{R}$ of class $C^2$, such that $b( \mathbf{x}, \mathbf{y}) = a( \mathbf{x}, \mathbf{y}) D_x H( \mathbf{x})$. Thus, further computation shows that \[
D_i H(\mathbf{x}) := \frac{b_i(\mathbf{x},\mathbf{y})}{a_{ii}(\mathbf{x},\mathbf{y})}= \frac{X_i(t) \beta_i(\mathbf{x},\mathbf{y})}{a_{ii}(\mathbf{x},\mathbf{y})} = \frac{C_x}{x_i},
\]
\[k(\mathbf{x},\mathbf{y}) := - \sum_{i=1}^n \sum_{j=1}^n \frac{a_{ij}( \mathbf{x}, \mathbf{y})}{2} [ D_{ij}^2 H( \mathbf{x}) + D_i H ( \mathbf{x},\mathbf{y}) D_j H( \mathbf{x})] = 0,\]
and the market price of risk follows
 $\theta_i(\mathcal{X}(t),\mathcal{Y}(t)) = s_{ii}'(\mathcal{X}(t),\mathcal{Y}(t)) D_i H(\mathcal{X}(t)) = C_x ( X(t))^{\frac{1}{2}} (X_i(t) \mathcal Y_i(t))^{-\frac{1}{2}}$, $i = 1, \ldots, n$, with $H( \mathbf{x}) = C_x \sum_{i=1}^n \log x_i$. $L(t) = \prod_{j=1}^n \frac{x_j}{X_j(t)}$, $\widehat{X}(t) = X(t)\prod_{j=1}^n \frac{x_j}{X_j(t)}.$

For $i = 1, \ldots, n$, the Fichera drift follows
\[
\begin{aligned}
    f_i(\mathbf{x}, \mathbf{y}) &= \frac{ a_{ii}(\mathbf{x}, \mathbf{y})}{\mathbf{x} \cdot \mathbf{1} + \mathbf{y} \cdot\mathbf{1}} - \frac{1}{2} D_i a_{ii} (\mathbf{x}, \mathbf{y}) \\
    &= \frac{x_i}{y_{0,i}} \frac{x_0}{\mathbf{x} \cdot \mathbf{1} + \mathbf{y} \cdot\mathbf{1}} - \frac{1}{2} \frac{x_0}{y_{0,i}} - \frac{1}{2} \frac{x_i}{y_{0,i}} \\
    &= \frac{x_i}{y_{0,i}} \left( \frac{x_0}{\mathbf{x} \cdot \mathbf{1} + \mathbf{y} \cdot\mathbf{1}} - \frac{1}{2} \right)- \frac{1}{2} \frac{x_0}{y_{0,i}} 
\end{aligned}
\]
and $f_i(\mathbf{x}, \mathbf{y}) = 0$, for $i = n+1, \ldots, 2n$. Since $ \frac{x_0}{\mathbf{x} \cdot \mathbf{1} + \mathbf{y} \cdot\mathbf{1}} - \frac{1}{2}<\frac{1}{2}$, we have $f_i(\mathbf{x}, \mathbf{y}) \leq 0$. This tells us that $\widehat{\mathcal V}(\cdot)$ and $\widehat X(\cdot)$ are strict local martingales. 
{Note that the corresponding auxiliary process $\zeta(\cdot)$ is modified to take values in $\mathbb R^n$, whose coefficients take the values of the first $n$ coordinates in the original process in Definition~\ref{aux}.}

We can now simplify \eqref{anotherpi}. 
With $u^{\cdot}(T-t, \mathbf{x}, \mathbf{y}) = u^{\cdot}(T-t)$ it holds
\begin{equation}
\label{eq: ex_optpi}
\begin{aligned}
\pi^{\ell \star}_i(t) &= 1 - (1 - \mathbf{m}_i(t))\frac{\delta X(t)}{\mathcal V_t} + X_i(t) D_i \log u^{\ell}(T-t) + \frac{1-\delta}{N \delta}\mathbf{m}_i(t) \sum_{k=1}^N V^k(t) D_i \log u^{k}(T-t) \\
& = \frac{\mathbf{m}_i(t) - \frac{(1 - \delta) X_i(t) }{N} \sum_{k=1}^N u^{k}(T-t) D_i \log u^{k}(T-t) - \frac{1-\delta}{N} \sum_{k=1}^N u^{k}(T-t)}{1- \frac{1-\delta}{N} \sum_{k=1}^N u^{k}(T-t) } + X_i(t) D_i \log u^{\ell}(T-t),
\end{aligned}
\end{equation}
where for $\ell = 1, \ldots, N$, by \eqref{ggg} and \eqref{ggg1},
\begin{equation}
\label{eq: uexample}
\begin{aligned}
u^{\ell}(T-t) = &\frac{e^{c_{\ell}} X_1(t) \cdots X_n(t)}{\mathcal{V}(t)}\E \left[\frac{\mathcal{V}(T)}{X_1(T) \cdots X_n(T)} \Big| \mathcal{F}_t \right]\\
 = &\frac{\delta e^{c_{\ell}} X_1(t) \cdots X_n(t)}{\mathcal{V}(t) (1-(1-\delta)\overline{e^c})}\E \left[\frac{X(T)}{X_1(T) \cdots X_n(T)} \Big| \mathcal{F}_t \right].
\end{aligned}
\end{equation}
As a special case, if $\delta=1$, $\tau(\mathbf{x},\mathbf{y}) = 0$, for any $(\mathbf{x},\mathbf{y}) \in (0, \infty)^n \times (0, \infty)^n$, then
\begin{equation}
\label{eq: pispecial}
    \pi^{\ell \star}_i(t) = X_i(t) D_{x_i} \log u^{\ell}(T-t, \mathbf{x}, \mathbf{y}) + \mathbf{m}_i(t).
\end{equation}

{This optimal strategy resembles the replicating strategy result in \cite[Equation (11.1)]{ra2005}. In fact, we recover the model in \cite{ra2005}, when $N = 1$ and there is no interaction term $\mathcal{Y}(\cdot)$ in the market model. In our case, $D_i \log u^{\ell}(T-t, \mathbf{x}, \mathbf{y})$ contains the interaction term $\mathbf{y}$, which influences the coefficients in the associated Cauchy problem of $u^{\ell}(T-t, \mathbf{x}, \mathbf{y})$ in \eqref{aineq}. Here, the instantaneous growth rates and volatility coefficients can be rewritten in the form of \eqref{eq: umn}, that is, functions of the market portfolio $\mathbf{m}(\cdot)$ and the relative weights of the trading volume $\mathbf{n}(\cdot) := (\mathbf{n}_1(\cdot), \ldots, \mathbf{n}_n(\cdot))$, $\mathbf{n}_i(\cdot) := \frac{\mathcal{Y}_i(\cdot)}{Y(\cdot)}$.
Then, the first term of \eqref{eq: mn} takes the form
\begin{equation}
\label{eq: mn}
\begin{aligned}
X_i(t) D_{x_i} \log u^{\ell, \text{\boldmath$\pi$}^{\star}}(T-t, \mathbf{x}, \mathbf{y}) &= \mathbf{m}_i \left(D_{m_i} \log \widetilde{U}(T-t, \mathbf{m}(t), \mathbf{n}(t)) - \sum_{j=1}^n \mathbf{m}_j D_{m_j} \log \widetilde{U}(T-t, \mathbf{m}(t), \mathbf{n}(t))\right)
\end{aligned}
\end{equation}
where $\widetilde{U}(\cdot) \in C^{1,3,3}((0, \infty) \times \Delta_n \times \Delta_n)$. Summing up \eqref{eq: pispecial} over all $i$, we get $\sum_{i=1}^n \pi^{\ell \star}_i(t) = 1$, for $t \in [0,T]$, $\ell = 1, \ldots, N$.}

If $c_{\ell} = c$, $\tau(\mathbf{x},\mathbf{y}) = 0$,  for any $(\mathbf{x},\mathbf{y}) \in (0, \infty)^n \times (0, \infty)^n$, then
\[
\pi^{\ell \star}_i(t) = \frac{X_i(t)}{1-(1-\delta)  u^{\ell}(T-t, \mathbf{x}, \mathbf{y})} D_i \log u^{\ell}(T-t, \mathbf{x}, \mathbf{y}) + \mathbf{m}_i(t).
\]
\end{example}

\begin{remark}
\label{cpdfs}
Theorem~\ref{thm: neforxy} provides the fixed-point characterization and its connection with functionally generated portfolios, but it does not yield a fully explicit closed-form solution for the optimal strategies. 
One can ensure the existence of their relative arbitrage opportunities through the Fichera drift to optimize the initial wealth in the sense of Nash equilibrium. In particular, if \eqref{anotherpi2} satisfies the condition about the Fichera drift $f_i(\cdot) < 0$ on each face of the boundary $\mathcal{O}^{2n}$ in Proposition~\ref{uniqpf} then there is a relative arbitrage opportunity for each investor. However, whether the Fichera drift condition is satisfied or not does not affect the achievement of Nash equilibrium.  

The analysis of the next section~\ref{sec: uniqueNE} is suitable not only for the particular functionally generated portfolio structure we mentioned above. Once we solve the fixed-point problem to derive the Nash equilibrium in Theorem~\ref{thm: neforxy}, we can solve the market coefficients of the given strategies through \eqref{eq:defy}. This is further explained in Appendix~\ref{apdx:fp}.
\end{remark}

\subsection{The uniqueness of Nash equilibrium}
\label{sec: uniqueNE}

Theorem~\ref{thm: neforxy} constructs a set of optimal strategies as the solution of a fixed point problem, and Example~\ref{ex: vsm} provides a case that the optimal strategy solved in Theorem~\ref{thm: neforxy} amounts to the Nash equilibrium. A natural question is the uniqueness of the Nash equilibrium in $N$-player games, since the unique Nash equilibrium leads to the unique optimal strategy for investors in the form of \eqref{anotherpi2}. To derive the uniqueness result, 
we consider the fixed point problem regarding the optimal arbitrage quantity:
With a given set of $\{(u^1, \ldots, u^{N})(T-t, \mathbf x, \mathbf y) \}_{t \in [0,T]}$, investors solve optimal strategies \eqref{anotherpi2}, and we solve the drift and diffusion coefficients of the coupled system $\mathcal{X}(t), \mathcal{Y}(t)$ through optimal strategies. Then we expect that the optimal arbitrage quantity based on this updated system coincides with the given $\{(u^1, \ldots, u^{N})(T-t, \mathbf x, \mathbf y) \}_{t \in [0,T]}$ we started with. In particular, at the Nash equilibrium, \eqref{optvstarequ} holds and 
\begin{equation}
\label{ute2}
    u^{\ell}(T-t,\mathbf{x}, \mathbf{y}) = e^{c_{\ell}} \mathbb{E}^{\mathbf{x}, \mathbf{y}} \big[ \mathcal{V}(T-t) L(T-t) \big]\, /\, \mathcal{V}(0) = \frac{\delta e^{c_{\ell}}}{\mathcal{V}(0)} \E^{\mathbf{x}, \mathbf{y}} \left[ \frac{\widehat{X}(T-t)}{1- \frac{1-\delta}{N}  U(t, \mathbf{x}, \mathbf{y}) \sum_{k=1}^N e^{c_{k}}} \right],
\end{equation}
where we define the corresponding common factor 
\begin{equation}
\label{eq: commonfactor}
    U(T-t, \mathcal{X}(t), \mathcal{Y}(t)) = u^{k}(T-t, \mathcal{X}(t), \mathcal{Y}(t))/e^{c_{k}}, \quad \text{for every } k = 1,\ldots, N.
\end{equation}
Here, $\widehat{X}(T-\cdot)$ depends on $U(\cdot)$ through \eqref{xhat}. 


%

\begin{prop}
    \label{cor: c}
    {
If $c_{\ell} \equiv c >0$ for $\ell = 1, \ldots, N$, and if the unique equilibrium $(\pi^\star, \mu^\star)$ is achieved, then the initial benchmark follows
\begin{equation}
\label{eq:samec_v}
    \mathcal{V}(0) = \frac{\delta x_0}{1-(1-\delta) u(T, \mathbf{x}, \mathbf{y})},    
\end{equation}
where $u(T, \mathbf{x}, \mathbf{y}) = e^c \E^{\mathbf{x}, \mathbf{y}}[\mathcal{V}(T)L(T)]/\mathcal{V}(0)$, for $\ell = 1, \ldots, N$.}
\end{prop}

 {By Proposition~\ref{prop: clprop}, $u^{\ell}(\cdot)$ is determined by the preference level $c_\ell$, the initial condition $(\mathbf{x}, \mathbf{y})$, and the discounted benchmark $\mathcal{V}(\cdot)L(\cdot)$.
If each investor seeks strategies according to Theorem~\ref{thm: neforxy}, then the optimal arbitrage quantity of each investor on the same time horizon is indistinguishable up to the preference level $c_\ell$. Thus, when $c_{\ell} = c$, every investor has the same optimal arbitrage objective $u(\cdot)$, and \eqref{eq:samec_v} follows from \eqref{v0v0} and \eqref{ute2}.}

To this end, we first summarize the procedure for arriving at a fixed-point solution in the space of the paths of $\{u^{\ell}(T-t, \mathcal{X}(t), \mathcal{Y}(t))\}_{t \in [0,T]}$, for every $\ell = 1,\ldots, N$ in the following chart.
\begin{figure}[ht] 
  \centering 
\begin{tikzpicture}[node distance = 4.5 cm, auto]
  \node (A) {$\{u^{\ell}(T-\cdot, \mathbf{x}, \mathbf{y})\}_{\ell=1}^N$};
  \node (B) [right of=A] {$\{\phi^{\ell}_t(\cdot)\}_{\ell=1}^N$};
  \node (C) [right of=B] {$(b,\sigma,\gamma,\tau)(\cdot) $};
  \node (D) [right of=C] {$ (\mathcal{X}, \mathcal{Y})$};
  
  \draw[->] (A) -- (B) node[midway] {Theorem~\ref{thm: neforxy}};
  \draw[->] (B) -- (C) node[midway] {\eqref{gonc1}};
  \draw[->] (C) -- (D) node[midway] {\eqref{eq: x}-\eqref{eq: y}, \eqref{eq:samec_v}};
  \draw[->,bend left=12] (D) to node[midway, fill=white] {\eqref{xhat}, \eqref{ute2}} (A);
\end{tikzpicture}
  \caption{The formulation of the fixed point problems. Note that this chart works for the fixed-point problem on the space of the paths of strategies as well, i.e., $\pi = \Phi(\pi)$ in Section~\ref{sec4.2.1}-\ref{sec4.2.2}, if we start the flow from $\{\phi^{\ell}(\mathbf{x},\mathbf{y})\}_{\ell=1}^N$.} 
  \label{fig: fixedpt} 
\end{figure}

Next, we provide the sufficient condition for the unique Nash equilibrium in the sense of Definition~\ref{def: uni-ne}. 
We consider the optimal arbitrage quantity $u(\cdot) \in C^{1,3,3}([0,T] \times \mathbb{R}^n \times \mathbb{R}^n)$ on the path space. Let $\mathcal{U} = C_b([0,T] \times \mathbb{R}^ n_+ \times \mathbb{R}^n_+; \mathbb{R})$ be the set of continuous $\mathbb R_+$-valued functions equipped with
the supremum norm 
\[\|\cdot\|_\mathcal U := \sup_{\tilde\tau \in [0,T], (\mathbf{x}, \mathbf{y}) \in \mathbb{R}^n_+\times \mathbb{R}^n_+} |u(\tilde \tau, \mathbf{x}, \mathbf{y})|\]

We make the following assumption for the uniqueness of the Nash equilibrium. 
\begin{asmp}
\label{asmp:decouplex}
    For $t \in [0,T]$, $\tilde\tau := T-t$, denote the deflated market capitalization in \eqref{xhat} as $\widehat{X}^u(\tau)$ when \eqref{eq: commonfactor} is $u(\tau, \cdot, \cdot)$. We assume that for every $u, v \in \mathcal U$, $t \in [0,T]$, it satisfies 
    \begin{equation}
    \sup_{\tilde\tau \in [0,T], \mathbf{x},\mathbf{y} \in \mathbb R^n_+ \times \mathbb R^n_+ }\E^{\mathbf{x}, \mathbf{y}} \left \lvert \widehat{X}^u(\tilde\tau) - \widehat{X}^v(\tilde\tau) \right \rvert < M \lVert u-v\rVert_{\mathcal U} 
    \end{equation}
    for some constant $M > 0$. 
\end{asmp}

One special case is when $\mathcal{X}(t)$ is of the form
\[
dX_i(t) = X_i(t)(\beta_i(\mathcal{X}(t)) dt + \sum_{k=1}^n \sigma_{ik}(\mathcal{X}(t)) dW_k(t)),\ \  i= 1, \ldots, n.
\] 
That is, the dynamics of the stock capitalization is not influenced by the trading volume of the investors. Hence, Assumption~\ref{asmp:decouplex} is satisfied, as the market is not influenced by investors, while the wealth processes of the investors are influenced by their empirical mean of the trading volume. 

\begin{thm}[Uniqueness of Nash equilibrium]
\label{thm:unique}
Under Assumption~\ref{asmp1:main}-\ref{asmp:decouplex}, consider the subproblems $u^{\ell}(T-t)$ at every starting time $t \in [0,T]$. For an arbitrary $t \in [0,T)$, take $u, v \in (0,1)$ as the different values of the initial relative arbitrage quantity of investor $\ell$, as defined in \eqref{ggg}. Nash equilibrium $(\text{\boldmath$\pi$}^{\star}, \mu^{\star})$ is unique when 

\begin{equation}
\label{eq: condforn}
\frac{1-\delta^2}{\delta} \overline{e^c}  \in (0,1),
 \quad M < x_0 \frac{\delta + \overline{e^{c}}(\delta^2 -1)}{1-(1-\delta) \overline{e^{c}}},
\end{equation}
$M$ is the constant in Assumption~\ref{asmp:decouplex}.
\end{thm}

\begin{proof}
\begin{enumerate}
\item  We first formulate the fixed-point problem of $U(\cdot)$.  
Define $f : \mathcal{U} 
\rightarrow \mathcal{U}$ as following:
\begin{equation}
\label{fu}
f(U) = \left[1- \frac{1-\delta}{N} \cdot U \cdot \sum_{k=1}^N e^{c_{k}}\right]^{-1}.
\end{equation}
Then, it holds,
\begin{equation}
\label{eq: d}
\begin{aligned}
    \sup_{\tilde\tau \in [0,T], (\mathbf{x}, \mathbf{y}) \in \mathbb{R}^n_+\times \mathbb{R}^n_+} f(U)(\tilde\tau,\mathbf{x}, \mathbf{y}) &= \left[1- \frac{1-\delta}{N} \sup_{\tilde\tau \in [0,T], (\mathbf{x}, \mathbf{y}) \in \mathbb{R}^n_+\times \mathbb{R}^n_+} U(\tilde\tau,\mathbf{x}, \mathbf{y}) \sum_{k=1}^N e^{c_{k}}\right]^{-1}\\
    & = \left[1- \frac{1-\delta}{N} \sum_{k=1}^N e^{c_{k}}\right]^{-1} =: C_f,
\end{aligned}
\end{equation}

Define an operator $\mathcal{F}: \mathcal{U} \rightarrow \mathcal{U}$ by
\begin{equation}
 \label{ufixed0}
\left[ \mathcal{F} * U \right](T-t,\mathbf{x}, \mathbf{y})= \mathcal{V}^{-1}(0)\E^{\mathbf{x}, \mathbf{y}} \left[\frac{\widehat{X}(T-t)}{1- \frac{1-\delta}{N}  U(t,\mathbf{x}, \mathbf{y}) \sum_{k=1}^N e^{c_{k}}} \right] = \mathcal{V}^{-1}(0)\E^{\mathbf{x}, \mathbf{y}} \left[ \widehat{X}(T-t) f (U)(t,\mathbf{x}, \mathbf{y}) \right],
\end{equation}
for every $t \in [0,T]$, where $U := \left(U_t\right)_{t \in [0,T]} \in \mathcal{U}$, $f$ is defined in \eqref{fu}. Consider a mapping $\mathcal{I}: \mathcal{U} \rightarrow \mathcal{U}$, such that $\mathcal I(U)(t, \mathbf{x}, \mathbf{y}) = U(T-t, \mathbf{x}, \mathbf{y})$, for every $t \in [0,T]$. In particular, $\mathcal{I}(\cdot)$ maps the optimal arbitrage quantity for the time horizon $[0,T-t]$ to the optimal arbitrage quantity for the time horizon $[0, t]$. Note that this is different from the subproblems in \eqref{utobj1}, since the time horizon there is $[t,T]$, for every $t \in [0,T]$. We have
\[
\mathcal I(\E^{\mathbf{x}, \mathbf{y}}[\widehat{X}])(T-t,\mathbf{x}, \mathbf{y}) = \E^{\mathbf{x}, \mathbf{y}}[\widehat{X}(t)].
\]
This mapping $\mathcal{I}$ is continuous and bounded in $\mathcal{U}$. The boundedness is immediately followed by the bounded nature of $U(\cdot)$. Let $U_m(\cdot)$ be a sequence of functions in $\mathcal{U}$ that converges uniformly to a function $U \in \mathcal{U}$ on $[0,T] \times \mathbb{R}_+ \times \mathbb{R}_+$ as $m \to \infty$. Then for each $t \in [0,T]$, $\mathcal I$ is continuous since
    \[\lim_{m \to \infty} \mathcal I(U_m)(t) = \lim_{m \to \infty} U_m(T-t) = U(T-t) = \mathcal I(U)(t)\]
    Thus, at Nash equilibrium, \eqref{ute2} leads to the fixed-point condition
\begin{equation}
\label{ufixed}
\left[ \mathcal{F} * \mathcal{I}(U) \right](T-t,\mathbf{x}, \mathbf{y}) = U(T-t,\mathbf{x}, \mathbf{y})
\end{equation}
for every $t \in [0, T]$, $(\mathbf{x}, \mathbf{y}) \in \mathbb{R}_+ \times \mathbb{R}_+$. 
\item Denote $\widehat{M} :=  M C_f \left(x_0 + \frac{1-\delta}{\delta} \widebar{v}_0\right)^{-1} = \frac{M}{\delta x_0}$, $\lambda := \sup_{\mathbf{x}, \mathbf{y}}  \E^{\mathbf{x}, \mathbf{y}} \left[\widehat{X}(T) \right]/\mathcal{V}(0)$, $L := \sup_{u \in [0,1]} |f'(u)| = \frac{1-\delta}{\left(1-(1-\delta) \overline{e^{c}}\right)^2} \overline{e^{c}}$. By the triangle inequality, for $u, v \in \mathcal{U}$,
\[
\begin{aligned}
 \|\mathcal{F} * \mathcal{I}(u)  - \mathcal{F} * \mathcal{I}(v) \|_{\mathcal U} & = \sup_{\tilde\tau \in [0,T], (\mathbf{x}, \mathbf{y}) \in \mathbb{R}^n_+\times \mathbb{R}^n_+} |\mathcal{F} * \mathcal I(u(\tilde\tau,\mathbf{x}, \mathbf{y})) - \mathcal{F} * \mathcal{I}(v(\tilde\tau,\mathbf{x}, \mathbf{y}))| \\
 &\leq \sup_{\tilde\tau, \mathbf{x}, \mathbf{y}} \E^{\mathbf{x}, \mathbf{y}} \left[\widehat{X}^u(\tilde\tau) \left | f(u) - f(v)\right | (t,\mathbf{x}, \mathbf{y})\right] / \mathcal{V}(0) + \widehat{M} \sup_{\tilde\tau, \mathbf{x}, \mathbf{y}} \E \left |\widehat{X}^u(\tilde\tau) - \widehat{X}^v(\tilde\tau) \right|\\
& \leq L \sup_{\tilde\tau, \mathbf{x}, \mathbf{y}} \E^{\mathbf{x}, \mathbf{y}} \left[\widehat{X}^u(\tilde\tau) |u - v|(\tilde\tau,\mathbf{x}, \mathbf{y}) \right] / \mathcal{V}(0) + \widehat{M} \|u - v\|_{\mathcal U}\\
& \leq \left( \lambda L + \widehat{M} \right) \|u - v\|_{\mathcal U},
\end{aligned}
\]
where the second inequality is derived from the local Lipschitz continuity of $f$.  
Since $\E \left[\widehat{X}(t) \right] < x_0$ from the supermartingale property proved in Proposition~\ref{f1}, we have
\[
\lambda := \sup_{\tilde\tau, \mathbf{x}, \mathbf{y}}  \E^{\mathbf{x}, \mathbf{y}} \left[\widehat{X}(\tilde\tau) \right]/\mathcal{V}(0) = \sup_{\tilde\tau, \mathbf{x}, \mathbf{y}}  \frac{\E^{\mathbf{x}, \mathbf{y}} \left[\widehat{X}(\tilde\tau) \right]}{X(0)} \frac{X(0)}{\mathcal{V}(0)} < \frac{X(0)}{\mathcal{V}(0)} = \frac{1-(1-\delta)\overline{e^{c}}}{\delta }.
\]


Combining these quantities, we get
\begin{equation}
\label{lambdal}
\begin{aligned}
    \lambda L & \leq  \frac{1}{ \delta } \frac{(1-\delta) \overline{e^{c}}}{1-(1-\delta)\overline{e^{c}}}. \\
\end{aligned}    
\end{equation}
\item To show that $ \mathcal{F} * \mathcal{I}$ is a contraction, it suffices to verify that $\lambda L + \widehat{M} < 1$. This is equivalent to
\[
 \overline{e^{c}} (1-\delta)< 1 , 
\]
and
\[
\frac{1}{\delta} \frac{(1-\delta) \overline{e^{c}}}{1-(1-\delta)\overline{e^{c}}} + \frac{M}{\delta x_0}< 1.
\]

Further computation yields that \eqref{eq: condforn} is a sufficient condition that every participant achieves the unique Nash equilibrium as derived in Proposition~\ref{prop: clprop}. 
Therefore, for all the functions $ u, v \in \mathcal{U}$, we conclude the contraction condition that there exists a $ 0 \leq k < 1 $ such that 
\[
\|\mathcal{F} * \mathcal{I}(u)  - \mathcal{F} * \mathcal{I}(v) \|_{\mathcal U}  \leq k \|u - v\|_{\mathcal U}.
\]
Then by the Banach Fixed Point Theorem and Assumption~\ref{asmp: geqn}, the solution $u(\cdot)  \in C^{1,3,3}([0,T] \times \mathbb{R}^ n_+ \times \mathbb{R}^n_+)$ in \eqref{ufixed} is unique. 
The optimal wealth processes are thus uniquely determined by 
\begin{equation} 
\label{vfix}
   V^{\ell\star}(t) 
= \frac{ u^{\ell}(T-t, \mathcal X (t), \mathcal Y(t)) \delta X(t)}{1-(1-\delta)\frac{1}{N}\sum_{\ell=1}^N  u^{\ell}(T-t, \mathcal X (t), \mathcal Y(t))},
\end{equation}
which indicates that the Nash equilibrium is unique in the sense of Definition~\ref{def: uni-ne}.

\end{enumerate}
\end{proof}

\begin{remark}
Under the uniqueness conditions for Nash equilibrium and the Fichera drift condition on the market coefficients, investors can outperform their benchmark and achieve optimal arbitrage at the unique Nash equilibrium. Consider Example~\ref{ex: vsm}, with a sufficiently small $T$ and a suitable value for $C_x$ in \eqref{eq: coeff-vsm}, we can satisfy the bound on $M$ in \eqref{eq: condforn}.
\end{remark}
\begin{remark}
We use \eqref{eq: ex_optpi} to show a counterexample for the two-fund separation theorem (\cite{tobin1958liquidity}). That is, the optimal strategy of the form \eqref{anotherpi} for $\delta \in (\delta_-, \delta_+)$ need not be a linear combination of the optimal strategy in the form \eqref{anotherpi} for $\delta_-$ and $\delta_+$. We know that a simultaneous relative arbitrage opportunity is not possible when $\delta =0$. When $\delta$ approaches zero, the conditions for achieving a unique Nash equilibrium are violated, adding uncertainties and difficulties to solve an optimal strategy for competitive investors.
\end{remark}

\section{Discussions and Future work}
\label{sec: future}

{In this paper, we discuss the optimization and game-theoretic analysis of the relative arbitrage problem with a finite number of competitive investors. In contrast to the standard single-player construction, where the market evolves along a given stochastic process, here investors need to consider how the market and other investors react to their trading strategies. To this end, we summarize several financial interpretations from this paper. First, when the benchmark used by every player is the average wealth, the relative arbitrage opportunity vanishes. To see this clearly, consider a simplified case when every investor wants to outperform exactly the average wealth, then the benchmark is also the market portfolio. So as opposed to single investor case who tries to outperform the market when the market is not influence by the decision of the investors, here, the market coincides with the average wealth; consequently, the equilibrium is arrived when all investors adopts this same average wealth strategy. Second, we see the connection between the existence of relative arbitrage opportunities and the Fichera drift condition. One promising direction for future work is to derive specific characterizations of Proposition 3.4 to identify the conditions for arbitrage opportunities. Third, we derive the optimal feedback strategies under Nash equilibrium. This provides us with a comprehensive understanding of the microstructure of the market, the existence of relative arbitrage opportunities, and interesting connections with other important topics in SPT such as functional generated portfolios and volatility stabilized market model.}

{This paper leads to a few future topics related to the formulation of the market, different notions of equilibrium, and short-term arbitrage opportunities. The structure of the market dynamics and the information structure of strategies lead to different characterizations of optimal arbitrage quantity and Nash equilibrium.
We may consider investigating the supply-demand mechanism and considering model uncertainty to relax the structural assumptions imposed in the paper. Furthermore, investors are considered fully competitive and symmetric up to the preference level $c_{\ell}$ in this paper. Extending this discussion to cooperative behavior among some investors and optimization with the knowledge of exogenous noise traders can be promising: How do we formulate a control problem here? What is the suitable notion of optimization for relative arbitrage opportunities? In addition, it is of practical use to investigate the numerical solution of the Cauchy problem, the differential games, and the corresponding functionally generated portfolio.
The numerical scheme for the single investor case is studied in \cite{thesis, yang2024finding} through the time-changed Bessel bridges, which is closely related to the capitalization of the stocks.  }

{Beyond SPT-specific applications and open-questions, the game-theoretic framework discussed in this paper leads to several interesting future directions in multi-agent optimization. In particular, the regime (Section 4.2) of searching for Nash equilibrium and the construction of different fixed point problems based on the specific notion of uniqueness for Nash equilibrium can be used for multi-agent problems in many applications, especially when the interaction function in the states and in the objective can differ, with common noise presented.
Given that the mean-field game framework \cite{yang2023relative} provides a good approximation for finite-player games, one promising direction is to extend the numerical scheme for the $N$-player game we consider in this paper when $N$ is large. With the mean-field approximation error and the numerical error from fixed point problem further quantified, one may consider an environment with model uncertainty \cite{uncertain, yang2026pathwise}, or data-driven approaches via generative modeling \cite{liu2022deep, deng2024reflected} or reinforcement learning \cite{guo2019learning} to relax the structural setup of the market systems in the paper.}


\begin{thebibliography}{999}

\bibitem{bass}
R. F. Bass, E. A. Perkins, 
\emph{Degenerate stochastic differential equations with
H\"older continuous coefficients and super-Markov chains}.
Trans. Amer. Math. Soc. 355, 373-405, (2003).

\bibitem{visb}
 E. Bayraktar, Y.-J. Huang, Q. Song, 
\emph{Outperforming the market portfolio with a given probability}.
Ann. Appl. Probab. 22(4), 1465-1494 (2012).


\bibitem{lecturemfg}
R. Carmona,
\emph{Lectures on BSDEs, Stochastic Control, and Stochastic Differential Games with Financial Applications}.
SIAM Book Series in Financial Mathematics 1 (2016).

\bibitem{mfgbook}
R. Carmona, F. Delarue, 
\emph{Probabilistic Theory of Mean Field Games with Applications I: Mean Field Games with Common Noise and Master Equations}.
Volume 84 of Probability Theory and Stochastic Modelling, Springer, 2018.



\bibitem{de1990noise}
    J.B. De Long, A. Shleifer, L.H. Summers, R.J. Waldmann, 
    \emph{Noise trader risk in financial markets}. Journal of Political Economy, 
    J. Pol. Econ., 98 703--738 (1990) 


\bibitem{deng2024reflected}
  W. Deng, Y. Chen, N. T. Yang, H. Du, Q. Feng, R. T. Chen, \emph{Reflected Schr$\backslash$" odinger Bridge for Constrained Generative Modeling}. arXiv preprint arXiv:2401.03228, 2024.



\bibitem{evans2010partial}
    L.C. Evans, 
    \emph{Partial Differential Equations}. Graduate studies in mathematics, 
    American Mathematical Society (2010) 


\bibitem{opta}
  D. Fernholz, I. Karatzas,
  \emph{On Optimal Arbitrage}.
  Ann. Appl. Probab., 20 1179-1204 (2010). 

\bibitem{uncertain}
  D. Fernholz, I. Karatzas,
  \emph{Optimal Arbitrage under model uncertainty}.
  Ann. Appl. Probab, Vol. 21, No.6, 2191-2225 (2011).
  
  
\bibitem{mf}
R. Fernholz, I. Karatzas,
  \emph{Stochastic portfolio theory: A survey. In \textsc{Handbook of Numerical Analysis. Mathematical Modeling and Numerical Methods in Finance} (A. Bensoussan, ed.)}.
  89-168 (2009).
  

\bibitem{spt}
  R. Fernholz,
  \emph{Stochastic Portfolio Theory, \text{volume 48 of} Applications of Mathematics (New York)}.
Springer-Verlag, New York. Stochastic Modelling and Applied Probability. (2002).

\bibitem{ra2005}
R. Fernholz, I. Karatzas,
  \emph{Relative arbitrage in volatility-stabilized markets}.
  Ann. Finance 1, 149-177 (2005).
  
\bibitem{diversity}
  R. Fernholz, I. Karatzas, C. Kardaras,
  \emph{Diversity and relative arbitrage in equity market}.
  Finance \& Stochastics 9, 1-27 (2005).

\bibitem{volar}
R. Fernholz, I. Karatzas, J. Ruf,
  \emph{Volatility and arbitrage}.
Ann. Appl. Probab. 28 (1) 378-417 (2018).


\bibitem{exitm}
H. F\"ollmer, 
\emph{The exit measure of a supermartingale}.
Zeitschrift fu\"r Wahrscheinlichkeitstheorie und
verwandte Gebiete 21, 154-166 (1972). 


\bibitem{fichera}
A. Friedman,
\emph{Stochastic Differential Equations and Applications}.
Vol. I, Vol. 28 of Probability and Mathematical Statistics, Academic Press, New York (1975). 

\bibitem{mfg}
O. Gu\'eant, J.M. Lasry, and P.L. Lions,  
\emph{Mean field games and applications}.
In R. Carmona et al., editor, Paris Princeton
Lectures in Mathematical Finance IV, volume 2003 of Lecture Notes in Mathematics. Springer Verlag (2010).


\bibitem{guo2019learning},
  X. Guo, A. Hu, R. Xu, J. Zhang, \emph{Learning mean-field games}.
Advances in neural information processing systems, 32, 2019.


  





\bibitem{heath2000martingales}
D. Heath, M. Schweizer,
\emph{Martingales versus pdes in finance: an equivalence result with examples}.
Journal of Applied Probability, 37(4):947–957, 2000.

\bibitem{itkin2022ergodic}
D. Itkin, B. Koch, M. Larsson, J. Teichmann.
\emph{Ergodic robust maximization of asymptotic growth under stochastic volatility}. arXiv preprint arXiv:2211.15628, 2022.

\bibitem{itkin2021open}
D. Itkin, M. Larsson, 
\emph{Open markets and hybrid Jacobi processes}. arXiv preprint arXiv:2110.14046, 2021.

\bibitem{itkin2022robust}
D. Itkin, M. Larsson, 
\emph{Robust asymptotic growth in stochastic portfolio theory under long-only constraints}. Mathematical Finance, 32(1):114–171, 2022.

\bibitem{lee2003smooth}
J. M. Lee,
\emph{Introduction to Smooth Manifolds}. Springer, 2023.


\bibitem{liu2022deep}
  G. Liu, T. Chen, O. So, E. Theodorou, \emph{Deep generalized schr{\"o}dinger bridge}. Advances in Neural Information Processing Systems, Volume 35, 9374--9388, 2022.
  

\bibitem{russo2022stochastic}
F. Russo and P. Vallois, 
\emph{Stochastic calculus via regularizations}. volume 11. Springer Nature, 2022.

\bibitem{tobin1958liquidity}
J. Tobin, 
\emph{Liquidity preference as behavior towards risk. The review of economic studies}. 25(2):65–86, 1958.

 \bibitem{sde}
 B. K. {{\O}}ksendal,
\emph{Stochastic Differential Equations: An Introduction with Applications}.
Springer, Berlin. 6th edition. ISBN 9783642143946. (2010).

\bibitem{wong0}
S. Pal, T.K.L. Wong,
\emph{The geometry of relative arbitrage}.
Math. Financ. Econ. 10, 263-293 (2016).

\bibitem{pinsky}
M. Pinsky,
\emph{A note on degenerate diffusion processes}.
Theor. Probability. Appl. 14, 502-506 (1969).


\bibitem{pinsky1995positive}
R.G. Pinsky, \emph{Positive Harmonic Functions and Diffusion},
  {Cambridge Studies in Advanced Mathematics}, (1995). 




\bibitem{bmk}
 E. Platen,  and D. Heath,
  \emph{A Benchmark Approach to Quantitative Finance}.
Springer, Berlin, 2006. MR2267213

\bibitem{ruf2}
J. Ruf,
\emph{Optimal Trading Strategies Under Arbitrage}.
PhD thesis, Columbia University, New York, USA (2011). 

\bibitem{hedgear}
J. Ruf,
 \emph{Hedging under Arbitrage}.
Math. Financ. 23, 297-317 (2013).


\bibitem{sard1942measure}
A. Sard,
\emph{The measure of the critical values of differentiable maps}. Bulletin of the American Mathematical Society, 48 (12): 883–890 (1942).

 
\bibitem{strong}
W. Strong, J.-P. Fouque,
\emph{Diversity and arbitrage in a regulatory breakup model}.
Ann Finance 7, 349--374 (2011).

\bibitem{supthm}
D. W. Stroock, S. R. S. Varadhan,
\emph{On the support of diffusion processes with applications to the strong maximum principle}.
Proceedings of the Sixth Berkeley Symposium on Mathematical Statistics and Probability, Volume 3: Probability Theory, Berkeley, Calif., University of California Press, pp. 333–359 (1972).

\bibitem{supthm2}
D. W. Stroock, S. R. S. Varadhan,
\emph{Multidimensional diffusion processes}. 
Springer-Verlag, Berlin, 2007.

\bibitem{pc}
L. J. Van Mellaert,  and P. Dorato,
\emph{Numerical solution of an optimal control problem
with a probability criterion}.
IEEE Transactions on Automatic Control AC-17 543-546 (1972). 

\bibitem{villani2021topics}
C. Villani, 
\emph{Topics in optimal transportation}. Volume 58. American Mathematical Soc., 2021.


\bibitem{wong}
T.K.L. Wong,
\emph{Information geometry in portfolio theory}.
In Frank Nielsen (Ed.), Geometric Structures of Information, Springer (2019).

\bibitem{tkwfgp}
T.K.L. Wong,
\emph{Optimization of relative arbitrage}.
Ann. Finance 11 345--382 (2015).

\bibitem{thesis}
T. Yang,
\emph{Topics in relative arbitrage, stochastic games and high-dimensional PDEs}.
Ph.D. Dissertation, University of California, Santa Barbara (2021).

\bibitem{yang2023relative}
N. T. Yang, T. Ichiba, 
\emph{Relative arbitrage opportunities in an extended mean field system}.
arXiv preprint arXiv:2311.02690, 2023.

\bibitem{yang2024finding}
N. T. Yang, T. Ichiba,  
\emph{Finding the nonnegative minimal solutions of Cauchy PDEs in a volatility-stabilized market}  SIAM Journal on Financial Mathematics 16 SC76--SC87 (2025).

\bibitem{yang2026pathwise}
N. T. Yang, \emph{Pathwise Learning of Stochastic Dynamical Systems with Partial Observations}. arXiv preprint arXiv:2601.21860, 2026.




\end{thebibliography}

\newpage 

\begin{appendices}

\section{Market dynamics and conditions}
\label{defs}
\numberwithin{equation}{section}
This section recalls some properties of the market that are related to the existence of relative arbitrage.
\begin{definition}[Non-degeneracy and bounded variance]
\label{nd}
A market is a family $\mathcal{M}= \{X_1, \ldots, X_n\}$ of the $n$ stocks, each of which is defined as in \eqref{eq: x}, such that the matrix $\alpha(t)$ is nonsingular for every $t \in [0,\infty)$, a.s. The market $\mathcal{M}$ is called  \textit{nondegenerate} if there exists a number $\epsilon>0$ such that for $x\in \mathbb{R}^n$
\begin{equation*}
    \mathbb{P}(x \alpha(t) x' \geq \epsilon \|x\|^2, \forall t \in [0,\infty))=1, 
\end{equation*}
The market $\mathcal{M}$ has \textit{bounded variance} from above, if there exists a number $M>0$ such that for $x\in \mathbb{R}^n$
\begin{equation*}
    \mathbb{P}(x \alpha(t) x' \leq M \|x\|^2, \forall t \in [0,\infty))=1,
\end{equation*}
\end{definition}
We restate the non-degeneracy in a specific form in Assumption~\ref{asmp1:main} in order to show the existence of relative arbitrage.

\section{Proofs}
\label{racp}
\numberwithin{equation}{section}

\begin{proof}[Proof of Theorem~\ref{thm31}]

We first show some main steps of computing \eqref{aineq}. Plugging \eqref{ggg} in the above equations set and using the Markovian property of $g(\cdot)$ gives
\[
\begin{aligned}
 \frac{\partial u^{\ell}(t,  \mathbf{x}, \mathbf{y})}{\partial t} g( \mathbf{x}, \mathbf{y}) &= \mathcal{L}(u^{\ell}(t,  \mathbf{x}, \mathbf{y}) g( \mathbf{x}, \mathbf{y})) - \big( k( \mathbf{x}, \mathbf{y}) + \Tilde{k}( \mathbf{x}, \mathbf{y})\big) u^{\ell}(t,  \mathbf{x}, \mathbf{y}) g( \mathbf{x}, \mathbf{y}).\\
\end{aligned}
\] 
By expanding the above, it follows
\[
\begin{aligned}
     \frac{\partial u^{\ell}(t)}{\partial t} &= \frac{1}{2} \sum_{i,j=1}^n a_{ij}(\mathbf{x},\mathbf{y}) \bigg( D_{ij}^2 u^{\ell}(t) + 2 D_i u^{\ell}(t) \frac{D_j g(\mathbf{x},\mathbf{y})}{g(\mathbf{x},\mathbf{y})} 
     +  u^{\ell}(t) \frac{D_{ij} g(\mathbf{x},\mathbf{y})}{g(\mathbf{x},\mathbf{y})}\bigg)\\
     & + 2 \sum_{i,j=1}^n a_{ij}(\mathbf{x},\mathbf{y}) \bigg(D_{i} u^{\ell}(t) +  u^{\ell}(t) \frac{D_i g(\mathbf{x},\mathbf{y})}{g(\mathbf{x},\mathbf{y})} \bigg) D_j H(\mathbf{x},\mathbf{y})\\
     & + \frac{1}{2} \sum_{i,j=1}^n a_{ij} [ D_{ij}^2 H(\mathbf{x},\mathbf{y}) + 3 D_i H (\mathbf{x},\mathbf{y}) D_j H(\mathbf{x},\mathbf{y})] u^{\ell}(t)\\
     & + \frac{1}{2} \sum_{p,q=1}^n \psi_{pq}(\mathbf{x},\mathbf{y}) \bigg(D_{pq}^2 u^{\ell}(t) + 2 D_p u^{\ell}(t) \frac{D_q g(\mathbf{x},\mathbf{y})}{g(\mathbf{x},\mathbf{y})} 
     +  u^{\ell}(t) \frac{D_{pq} g(\mathbf{x},\mathbf{y})}{g(\mathbf{x},\mathbf{y})}\bigg)\\
     & + 2 \sum_{p,q=1}^n \psi_{pq}(\mathbf{x},\mathbf{y}) \bigg(D_{p} u^{\ell}(t) +  u^{\ell}(t) \frac{D_p g(\mathbf{x},\mathbf{y})}{g(\mathbf{x},\mathbf{y})} \bigg) D_q H(\mathbf{x},\mathbf{y})\\
     & + \frac{1}{2} \sum_{p,q=1}^n \psi_{pq} [ D_{pq}^2 I(\mathbf{y}) + 3 D_p H (\mathbf{x},\mathbf{y}) D_q H(\mathbf{x},\mathbf{y})]u^{\ell}(t)\\
     &+ \sum_{i,p=1}^n (s \tau^T)_{ip}(\mathbf{x}, \mathbf{y}) \bigg( D_{ip}^2 u^{\ell}(t) + D_i u^{\ell}(t) \frac{D_p g(\mathbf{x},\mathbf{y})}{g(\mathbf{x},\mathbf{y})} + D_p u^{\ell}(t) \frac{D_i g(\mathbf{x},\mathbf{y})}{g(\mathbf{x},\mathbf{y})} 
 +  u^{\ell}(t) \frac{D_{ip} g(\mathbf{x},\mathbf{y})}{g(\mathbf{x},\mathbf{y})}\bigg)\\
     & - \sum_{i,p=1}^n (s \tau^T)_{ip}(\mathbf{x}, \mathbf{y}) D_i H (\mathbf{x},\mathbf{y}) D_p H(\mathbf{x},\mathbf{y}) u^{\ell}(t).
\end{aligned}
\]
We can simplify this equation with the following computations.

By \eqref{calv}, and the definition of $g(\cdot)$,
\[
\frac{D_i g(\mathbf{x}, \mathbf{y})}{g(\mathbf{x}, \mathbf{y})} = -D_i H(\mathbf{x}, \mathbf{y}) + \frac{\delta}{\delta \mathbf{x} \cdot \mathbf{1} + (1-\delta) \mathbf{y} \cdot\mathbf{1}},\quad \frac{D_{p} g(\mathbf{x}, \mathbf{y})}{g( \mathbf{x}, \mathbf{y})} = -D_p H(\mathbf{x}, \mathbf{y})+ \frac{1-\delta}{\delta \mathbf{x} \cdot \mathbf{1} + (1-\delta) \mathbf{y} \cdot\mathbf{1}}.\] 
Th second order derivative with respect to $\mathbf{x}$ is
\[
\begin{aligned}
\frac{D_{ij} g(\mathbf{x}, \mathbf{y})}{g(\mathbf{x}, \mathbf{y})} &= - \frac{\delta (D_i H(\mathbf{x}, \mathbf{y})+D_j H(\mathbf{x}, \mathbf{y}))}{\delta \mathbf{x} \cdot \mathbf{1} + (1-\delta) \mathbf{y} \cdot\mathbf{1}}  - D_{ij}^2 H(\mathbf{x}, \mathbf{y}) + D_i H (\mathbf{x}, \mathbf{y}) D_j H(\mathbf{x}, \mathbf{y}),
\end{aligned}
\]
and the counterparts of second order derivative $\frac{D_{pq} g(\mathbf{x}, \mathbf{y})}{g(\mathbf{x}, \mathbf{y})}$ and $\frac{D_{ip} g(\mathbf{x}, \mathbf{y})}{g(\mathbf{x}, \mathbf{y})}$ can be derived in the same vein. As a result, when the drift term $\gamma(\cdot)$ and volatility term $\tau(\cdot)$ in \eqref{eq: y} is given, \eqref{eqsol} - \eqref{aineq} are satisfied. 

Suppose that a solution of \eqref{inequl} and \eqref{eqsol2} is $\Tilde{w}^{\ell}: C^{2}((0, \infty) \times (0, \infty)^n \times (0, \infty)^n) \rightarrow (0, \infty)$, $\Tilde{w}^{\ell}(0) = e^{c_{\ell}}$. Define $\Tilde{N}(t) :=  \Tilde{w}^{\ell}(T-t, \mathcal{X}_{t},  \mathcal{Y}_{t}) \mathcal{V}(t) L(t)$, $0 \leq t \leq T$. We solve 
\[
\frac{d \Tilde{N}(t)}{\Tilde{N}(t)} = \frac{d \Tilde{w}^{\ell}(T-t, \mathcal{X}_{t},  \mathcal{Y}_{t})}{\Tilde{w}^{\ell}(T-t, \mathcal{X}_{t},  \mathcal{Y}_{t})} + \big(\Pi^{\prime}(t) \sigma(t) -  \theta'(t) \big) \Big (1+ \frac{d \Tilde{w}^{\ell}(T-t, \mathcal{X}_{t},  \mathcal{Y}_{t})}{\Tilde{w}^{\ell}(T-t, \mathcal{X}_{t},  \mathcal{Y}_{t})} \Big )dW(t)
\]
by using 
the inequality \eqref{inequl}. We get that the $dt$ terms in $d \Tilde{N}(t)/\Tilde{N}(t)$ is always no greater than 0. $\Tilde{N}(t)$ is a positive local supermartingale. Thus, $\Tilde{N}(t)$ is a supermartingale.

Hence, $\Tilde{N}(0) = \Tilde{w}^{\ell}(T, \mathbf{x},  \mathbf{y}) \mathcal{V}(0) \geq \mathbb{E}^{\mathbb{P}}[\Tilde{N}(t)] = \mathbb{E}^{\mathbb{P}}[e^{c_{\ell}} \mathcal{V}(T) L(T)]$ holds for every $(T, \mathbf{x},  \mathbf{y}) \in (0,\infty) \times (0,\infty)^n \times (0,\infty)^n$. Then $\Tilde{w}^{\ell}(T, \mathbf{x},  \mathbf{y}) \geq \mathbb{E}^{\mathbb{P}} \big[ e^{c_{\ell}} \mathcal{V}(T) L(T) \big]/ \mathcal{V}(0) = u^{\ell}(T, \mathbf{x},  \mathbf{y})$.
\end{proof}

\begin{proof}[Proof of Proposition~\ref{f1}]
From It\^o's formula, the discounted process $\widehat{V}^{\ell}(\cdot)$ admits
\[
   d \widehat{V}^{\ell}(t) =  \widehat{V}^{\ell}(t)\big(\pi^{\ell\prime}(t) \sigma(t) -  \theta'(t) \big)dW(t); \quad  \widehat{V}^{\ell}(0) = \widehat{v}_{\ell}.
\]
Thus, $\widehat{V}^{\ell}(\cdot)$ is a positive local martingale with  $\mathbb{E}[L(T)] \leq 1$.
In particular, $\widehat{V}^{\ell}(\cdot)$ is a supermartingale, and we get that for an arbitrary $\omega^{\ell}$ in \eqref{utobj},
\[
\omega^{\ell} \mathcal{V}(0) = \widehat{v}_{\ell} \geq \mathbb{E} \big[\widehat{V}^{\ell}(T) \big] \geq \mathbb{E} \bigg[\widehat{X}(T)\delta e^{c_{\ell}} + L(T) (1-\delta) e^{c_{\ell}} \overline{V}(T)  \bigg] := p^{\ell},
\]
{where $\overline{V}(T) = \sum_{k\neq \ell} V^k(T) + V^\ell(T)$.}
Hence, $u^{\ell}(T, \mathbf x, \mathbf y) \geq \frac{p^{\ell}}{\mathcal{V}(0)}$, $\ell = 1, \ldots, N$.

To prove the opposite direction $u^{\ell}(T, \mathbf x, \mathbf y) \leq \frac{p^{\ell}}{\mathcal{V}(0)}$, we use the martingale representation theorem (Theorem 4.3.4, \cite{sde}) to find
\begin{equation}
U^{\ell}(t) := \mathbb{E}\big[ e^{c_{\ell}} \mathcal{V}(T) L(T)|\mathcal{F}_t \big]= e^{c_{\ell}}\int_0^t \Tilde{p}'(s) dW(s) + p^{\ell}, \quad 0 \leq t \leq T,
\end{equation} 
where $\Tilde{p}: [0,T] \times \Omega \rightarrow \mathbb{R}^n$ is ${\mathbb {F}}$-progressively measurable and almost surely square-integrable.
Next, construct a wealth process $V_*^{\ell}(\cdot) = U^{\ell}(\cdot)/L(\cdot)$, it satisfies $V_*^{\ell}(0) = p^{\ell}$, $V_*^{\ell}(T) = e^{c_{\ell}} \mathcal{V}(T)$. Use the trading strategy $h_*^{\ell}(\cdot)$ in \eqref{vlhatt}, where
\[
h_*^{\ell}(\cdot)  = \frac{1}{L(\cdot)V_*^{\ell}(\cdot)} \alpha^{-1}(\cdot) \sigma(\cdot) [ \Tilde{p}(\cdot) + U^{\ell}(\cdot) \theta(\cdot)].
\]
It follows that $h_*^{\ell} \in \mathbb{A}$ replicates $V_*^{\ell}(0)$, i.e., $V^{h_*^{\ell}}(T) = e^{c_{\ell}} \mathcal{V}(T)$ a.s., with $V^{h_*^{\ell}}(0) = p^{\ell}$. Consequently, from
\[
 \frac{p^{\ell}}{\mathcal{V}(0)} \in \big \{ \omega > 0 | \text{ there exists } \pi^{\ell} \in \mathbb{A}, \text{given} \  \pi^{-\ell}(\cdot) \in \mathbb{A}^{N-1}, \  \text{such that } V^{\omega \mathcal{V}(0), \pi^{\ell}} \geq e^{c_{\ell}} \mathcal{V}(T) \big \} 
 \]
it follows that $p^{\ell}/\mathcal{V}(0) \geq u^{\ell}(T)$. We therefore conclude $u^{\ell}(T) =  \mathbb{E} \big[ e^{c_{\ell}} \mathcal{V}(T) L(T) \big]\, /\, \mathcal{V}(0)$, for $\ell = 1, \ldots, N$.
\end{proof}


\section{General finite dynamical system}
\label{nmckean}

We use interacting particle models to describe the market. For fixed $N$, we model the $N$ investors as $N$ particles, where the particles have a common source of noise $W$ defined on $(\Omega, \mathcal{F},\mathbb{P})$. For any metric space $(\mathbb{X}, d)$, $\mathcal{P}(\mathbb{X})$ denotes the space of probability measures on $\mathbb{X}$ endowed with the weak convergence topology. $\mathcal{P}_p(\mathbb{X})$ is the subspace of $\mathcal{P}(\mathbb{X})$ of the probability measures of order $p$, that is, ${\bm \mu} \in \mathcal{P}_p(\mathbb{X})$ if $\int_{\mathbb{X}} d(x, x_0)^p {\bm \mu}(dx) < \infty$, where $x_0 \in \mathbb{X}$ is an arbitrary reference point. For $p \geq 1$, ${\bm \mu}$, ${\bm \nu}$ $\in \mathcal{P}_p(\mathbb{X})$, the $p$-Wasserstein metric on $P_p(\mathbb{X})$ is defined by
\[
W_p({\bm \nu}_1, {\bm \nu}_2)^p := \inf_{{\bm \pi} \in \Pi({\bm \nu}_1, {\bm \nu}_2)} \int_{\mathbb{X} \times \mathbb{X}} d(x,y)^p {\bm \pi}(dx,dy),
\]
where $d$ is the underlying metric on the space. $\Pi({\bm \nu}_1, {\bm \nu}_2)$ is the set of Borel probability measures ${\bm \pi}$ on $\mathbb{X} \times \mathbb{X}$ with the first marginal ${\bm \nu}_1$ and the second marginal ${\bm \nu}_2$. Specifically, ${\bm \pi}(A \times \mathbb{X}) = {\bm \nu}_1(A)$ and ${\bm \pi} (\mathbb{X} \times A) = {\bm \nu}_2(A)$ for every Borel subset $A \subset \mathbb{X}$. 

Let $C([0,T]; \mathbb{R}^{d_0})$ be the space of continuous functions from $[0,T]$ to $\mathbb{R}^{d_0}$. In this paper, we often take $\mathbb{X} = \mathbb{R}^{d_0}$ when considering a real-valued random variable or take $\mathbb{X}$ as the path space $\mathbb{X} = C([0,T]; \mathbb{R}^{d_0})$ with metric $d$ as the supremum norm for a process, where a fixed number $d_0$ will be specified later. 
$\mathcal{P}_p(\mathbb{R}^{d_0})$ equipped with the Wasserstein distance $\mathcal{W}_p$ is a complete separable metric space, since $\mathbb{R}^{d_0}$ is complete and separable. See \cite{villani2021topics} for more details on the Wasserstein metric and its properties.

We first define the empirical measure in the finite-particle system that we use in this paper.
\begin{definition}[Empirical measure in the finite $N$-particle system]
\label{nunt}
Consider $\mathcal{F}$-measurable $ C([0,T]; \mathbb{R}_+) \times C([0,T]; \mathbb{R}^n_+)$-valued random variables $(V^{\ell}, \pi^{\ell})$ for every investor $\ell = 1, \ldots, N$. We define the empirical measure of $(V^{\ell}, \pi^{\ell})$ as $\nu \in \mathcal{P}_2(C([0,T]; \mathbb{R}_+) \times C([0,T]; \mathbb{R}_+^n)) \cong \mathcal{P}_2(C([0,T]; \mathbb{R}_+ \times \mathbb{R}_+^n))$ , whose time-$t$ marginal is
\begin{equation}
\nu_t := \frac{1}{N} \sum_{\ell=1}^N {\bm \delta}_{(V^{\ell}(t), \pi^{\ell}(t))}, \quad  t \ge 0, 
\end{equation}
where ${\bm \delta}_x$ is the Dirac delta mass at $x \in \mathbb{R}_+ \times \mathbb{R}_+^n$. Thus, for any Borel set $A \subset \mathbb R_{+} \times \mathbb R_+^n$,
\begin{equation}
\nu_t(A) = \frac{1}{N} \sum_{\ell=1}^N {\bm \delta}_{(V^{\ell}(t), \pi^{\ell}(t))} (A)  = \frac{1}{N} \cdot \# \{\ell \leq N: (V^{\ell}(t), \pi^{\ell}(t)) \in A\} , 
\end{equation}
where $\# \{\cdot\}$ represents the cardinality of the set. In particular, the weighted average vector $\mathcal Y$ defined in \eqref{eq:defy} is given by $\mathcal Y(t) = \int_{\mathbb{R}_+ \times \mathbb R_+^n} x y \, \nu_t (dx \times d y) $, $t \ge 0 $, where $x$ represents wealth $V^{\ell}(t)$ and $y$ represents the strategies defined in the admissible set $\pi^{\ell}(t) \in \mathbb R_+^n$.
\end{definition}
 
Denote $\mathcal{X}(t) = (X_1(t), \ldots, X_n(t))$, $\mathbf{V}_t = (V^1(t), \ldots, V^{N}(t))$ for $t \ge 0$. For a fixed $N$, with $\nu^N_t$ in Definition~\ref{nunt} that generalizes $\mathcal{Y}(t)$, we can generalize the $(n+N)$-dimensional system as
\begin{equation}
\label{nparticle}
 d{X}_{i}(t) = {X}_{i}(t) \beta_{i}(t, \mathcal{X}(t), \nu_t)dt + \sum_{k=1}^{n}{X}_{i}(t) \sigma_{ik}(t, \mathcal{X}(t), \nu_t)dW_{k}(t); \quad \mathcal{X}(0) = \mathbf{x}_0
\end{equation}
for $i=1, \ldots , n$, and for $\ell = 1, \ldots , N$, 
\begin{equation}
\label{nv}
    dV^{\ell}(t) = V^{\ell}(t) \big( \sum_{i=1}^{n}\pi^{\ell}_i(t) \beta_{i}(t, X(t), \nu_t)dt + \sum_{i=1}^{n} \sum_{k=1}^{n} \pi^{\ell}_i(t) \sigma_{ik}(t, X(t), \nu_t)dW_{k}(t) \big); \quad V^{\ell}(0) = v^{\ell}.
\end{equation}

On the filtered probability space $(\Omega, \mathcal{F},\mathbb{P})$, we call
\[
(\mathcal{X}, \mathbf{V}, \nu, W) \in (C([0,T]; \mathbb{R}^n_+), C([0,T]; \mathbb{R}^N_+), \mathcal{P}_2(C([0,T]; \mathbb{R}_+ \times \mathbb{R}_+^n)), C([0,T]; \mathbb{R}^n))
\]
 {a strong solution} of the conditional McKean-Vlasov system \eqref{nparticle}-\eqref{nv} with respect to the filtration generated by the fixed $n$-dimensional Brownian motion $W(\cdot)$, with initial condition $(\mathcal{X}_0, \mathbf{V}_0, \nu_0)$, 
if $(X_t)_{t \in [0,T]}$ has continuous sample paths, satisfies $\mathbb P$-almost surely
\begin{equation}\label{sol}
    X_i(t) = X_i(0) + \int_0^t X_i(s)\beta_i(s, \mathcal X(s), \nu_s) \, {\rm d} s + \sum_{k=1}^n \int_0^t X_i(s) \sigma_{ik}(s, \mathcal X(s), \nu_s) \, {\rm d} W_k(t) , \qquad t \in [0,T]
\end{equation}
and is adapted to the smallest complete filtration $\mathbb F=(\mathcal F_t)_{t \geq 0}$ in which $X_0$ is $\mathcal{F}_0$-measurable and $W$ is $\mathbb F$-adapted. The system \eqref{eq: x} and \eqref{eq: y} is a special case of the above.

We make the following assumptions to ensure that the system \eqref{nparticle}-\eqref{nv} is well-posed. In the following, $|\cdot|$ denotes the Euclidean norm of vector $\mathbb{R}^d$ and the Frobenius norm of matrix $\mathbb{R}^{d\times n}$, $d =n$ or $N$ in particular. Also, let us define 
\begin{equation}\label{eq:b-and-s}
b_i(t, x, \nu) := x_i \beta_i(t, x, \nu), \quad s_{ik}(t, x, \nu) := x_i \sigma_{ik}(t, x, \nu); \quad 1 \le i \le n, t \ge 0 , x \in \mathbb R^{n}, \nu \in \mathcal P^{2}. 
\end{equation} 
{As detailed in the below assumption, the strategy is of a closed-loop feedback form. We consider it to be a function of the wealth processes. In Sections 4-5, we construct it as a function of the capitalization and the trading volume. If participants adopt open-loop strategies, the interaction with other players and the market is much limited, as their wealth depends on their own open loop control and a small part of the trading volume. This condition is desirable to impose for trading strategies, as the strategies will not be too volatile in the face of small changes in wealth. Furthermore, the Lipschitz continuity of market coefficients and strategy functions is common in the literature on mathematical finance and stochastic games, for example, in \cite[Lemma 3.3]{mfgbook}.}

\begin{asmp}
\label{asmp6:main}
\begin{enumerate}[label=\alph*.]
\item \label{asmp1alipfn2}
Assume the Lipschitz continuity and linear growth condition are satisfied with Borel measurable mappings $b_{i}(t, x, \nu)$, $s_{ik}(t, x, \nu)$ in \eqref{eq:b-and-s} from $[0,T] \times \mathbb{R}^n_+ \times \mathcal{P}_2(\mathbb{R}_+ \times \mathbb{R}_+^n)$ to $\mathbb{R}^n$. That is, there exists a constant $C_1, C_2 \in (0, \infty)$ that is independent of $t \in [0,T]$, such that
\begin{equation}
\label{bslip2}
    |b(t, x, \nu) - b(t, \widetilde{x}, \widetilde{\nu})| + |s(t, x, \nu) - s(t, \widetilde{x}, \widetilde{\nu})| \leq C_1[|x - \widetilde{x}| + \mathcal{W}_2(\nu, \widetilde{\nu})],
\end{equation}
\[
|x \beta(t, x, \nu)| + |x \sigma(t, x, \nu)| \leq C_2(1+|x|+M_2(\nu)),
\]
and
$a_{ij} (\cdot)$ satisfy the nondegeneracy condition, i.e., if there exists a number $\epsilon>0$ such that 
\[
    a_{ij} ( x, \nu) \geq \epsilon (|x|^2 + M_2^2(\nu)), ( \mathbf{x}, \mathbf{y}) \in \mathbb{R}_+^n \times \mathbb{R}_+^n.
\]
where 
\[
M_2(\nu) = \bigg(\int_{C([0,T]; \mathbb{R}_+ \times \mathbb{R}_+^n)} |x|^2 d \nu(x) \bigg)^{1/2}; \quad \nu \in \mathcal P_{2}(\mathbb R_{+} \times \mathbb R_+^n). 
\]

\item \label{growasmp}
Assume the following Lipschitz continuity and boundedness, i.e., there exist constants $C_2, B \in (0, \infty)$ such that 
\[
|v^{} \beta(t, x, \nu) - \widetilde{v}^{} \beta(t, \widetilde{x}, \widetilde{\nu})| + |v^{} \sigma(t, x, \nu) - \widetilde{v}^{} \sigma(t, \widetilde{x}, \widetilde{\nu})| \leq C_2[|x - \widetilde{x}| + n|v^{}-\widetilde{v}^{}| + \mathcal{W}_2(\nu, \widetilde{\nu})],
\]
\[
|v^{} \beta(t, x, \nu)| + |v^{} \sigma(t, x, \nu)| \leq B,
\]
for every $v, \widetilde{v} \in \mathbb{R}_+$, 
$t \in [0, T] $; $x, \widetilde{x} \in \mathbb R_{+}^n$; $\nu, \widetilde{\nu} \in \mathcal P_{2}( \mathbb{R}_+ \times \mathbb{R}_+^n)$. 
\item 
\label{asmp2:pipi}
Let $n$ be fixed and $\ell = 1, \ldots, N$. We assume the strategies adopted by investors are closed loop feedback controls of the wealth processes and are Lipschitz continuous in their variables, i.e., there exists a bounded mapping {$\phi^{\ell}: \mathbb{R}_+^N \to \mathbb{R}^n$} such that $\pi^{\ell}(t) = \phi^{\ell}(\mathbf{V}_t)$, $|\phi^{\ell}(\cdot)| < M$, and
\[
|\phi^{\ell}(\mathbf{v}) - \phi^{\ell}(\widetilde{\mathbf{v}})| \leq n C_3|\mathbf{v} - \widetilde{\mathbf{v}}| 
\]
for every $\mathbf{v}, \widetilde{\mathbf{v}} \in \mathbb R^{N}_{+}$. 
\end{enumerate}
\end{asmp}

\begin{thm}
\label{eufinite}
Assume that the stock capitalization vector $\mathbf{x}_0$ at time $0$ has a finite second moment, that is, $\E \lvert \mathbf{x}_0 \rvert^2 < \infty$, and is independent of the Brownian motion $W (\cdot)$.
Under Assumptions~\ref{asmp6:main}, the $(n+N)$-dimensional SDE system \eqref{nparticle}-\eqref{nv} admits a unique strong solution for any given number of stocks$n$, and any given number of investors $N$.
\end{thm}

\begin{proof}
We restrict the discussion to the time-homogeneous case, whereas the inhomogeneous case can be proved in the same fashion. Rewrite the system as a $(n+N)$-dimensional SDE system:
\begin{equation}
\label{xveqset}
\begin{aligned}
d \begin{pmatrix} 
   \mathcal{X}_t  \\
   \mathbf{V}_t  
   \end{pmatrix} 
   &:= f(\mathcal{X}(t),\mathbf{V}(t), \nu_t) dt + g(\mathcal{X}(t),\mathbf{V}(t), \nu_t)  dW_t, 
\end{aligned}
\end{equation}
where $f(\mathcal{X}(t),\mathbf{V}(t), \nu_t) :=  (f_1(\cdot), \ldots, f_{n+N}(\cdot))$, $f_i(\cdot)=X_i(t) \beta_i(\cdot)$ for $i=1,\ldots,n$, $f_j(\cdot):=V^{j-n}_t \pi^{j-n}_t \beta(\cdot)$ for $j=n+1,\ldots,n+N$; Similarly, $g(\mathcal{X}(t),\mathbf{V}(t), \nu_t) := (g_1(\cdot), \ldots, g_{n+N}(\cdot))$, $g_i(\cdot) := X_i(t) \sigma_i(\cdot)$ for $i=1,\ldots,n$ and  $g_j(\cdot) := V^{j-n}_\cdot \pi^{j-n}_\cdot \sigma(\mathcal{X}_\cdot,\nu_\cdot)$ for $j=n+1,\ldots,n+N$. 

Let us consider a closed-loop strategy $\pi^{\ell}_t = \phi^{\ell}(\mathbf{V}(t))$. Define a mapping $L_N : \mathbb{R}^{N}_+  \rightarrow \mathcal{P}_2(C([0,T]; \mathbb{R}_+ \times \mathbb{R}_+^n))$ 
\[
L_N(\mathbf{V}(t)) = \frac{1}{N} \sum_{\ell=1}^N {\bm \delta}_{(V^{\ell}(t), \phi^{\ell}(\mathbf{V}(t)))} = \nu_t.
\]
and define $F : \mathbb{R}_+^{N+n} \rightarrow \mathbb{R}^{N+n}$, $G : \mathbb{R}_+^{N+n} \rightarrow \mathbb{R}^{N+n} \times \mathbb{R}^{n}$, with
\[
F(\mathcal{X}(t),\mathbf{V}(t)) = f(\mathcal{X}(t),\mathbf{V}(t), L_N(\mathbf{V}(t))); \quad G(\mathcal{X}(t),\mathbf{V}(t)) = g(\mathcal{X}(t),\mathbf{V}(t), L_N(\mathbf{V}(t))).
\]
Write $(x,v) = (x_1, \ldots, x_n, v^1, \ldots, v^N)$ and $(y,u) = (y_1, \ldots, y_n, u^1, \ldots, u^N)$ for two pairs of random values of $(\mathcal{X}(\cdot),\mathbf{V}(\cdot))$. Denote the empirical measure $\Tilde{\pi}$ induced by the joint distribution of the random variables $u$ and $v$, i.e., 
\[
\Tilde{\pi} = \frac{1}{N} \sum_{\ell=1}^N {\bm \delta}_{(u^{\ell},v^{\ell})} . 
\] 
It is a coupling of the functions $L_N(v)$ and $L_N(u)$. From the definition of the Wasserstein distance, we have  
\begin{equation}
\label{w2ineq}
\begin{aligned}
 \mathcal{W}_2^2(L_N(v),L_N(u)) &\leq \int_{\mathbb{R}^{N} \times \mathbb{R}^{N}} |(v,\phi(v))-(u,\phi(u))|^2 \Tilde{\pi}(dv, du)\\
 &\leq \frac{1}{N} \sum_{\ell=1}^N |(v^{\ell},\phi^{\ell}(v))-(u^{\ell},\phi^{\ell}(u))|^2\\
 &\leq \Big (\frac{1}{N}+n^2C_3^2 \Big)|v-u|^2.
\end{aligned}
 \end{equation}


For every $(x,v)$ and $(y,u)$ in $\mathbb{R}_+^n \times \mathbb{R}^N_+$, and for every $\ell$, by \eqref{w2ineq}, the Cauchy-Schwartz inequality and the Lipschitz condition of $\beta_{i}$ and $\phi^{\ell}$, we have 
\[
\begin{aligned}
|F(x, v) - F(y, u)|^2 
 \le & \sum_{i=1}^n |b_i(x, L_N(v)) - b_i(y, L_N(u))|^2 + \sum_{\ell=1}^N |v^{\ell} \phi^{\ell}(v) \beta(x, L_N(v)) - u^{\ell} \phi^{\ell}(u) \beta(y, L_N(u))|^2 \\
 \leq & \,  2C_1^2 [ |x-y|^2 + \mathcal{W}_2^2(L_N(v),L_N(u))] \\
 &+ 6 M^2 C_2^2 [N|x-y|^2 + n^2|v-u|^2+ N \mathcal{W}_2^2(L_N(v),L_N(u))] + 2N n^2B^2 C_3^2 |v-u|^2\\ 
  \leq & \,  2C_1^2 [ |x-y|^2 + (\frac{1}{N}+n^2C_3^2)|v-u|^2] \\
 &+ 6N M^2 C_2^2 |x-y|^2 + 6 C_2^2 (n^2+n^2NC_3^2 + 1)|v-u|^2+ 2N n^2 B^2 C_3^2 |v-u|^2\\ 
\leq & \,   L_m^2  |(x,v)-(y,u)|^2,
\end{aligned}
\]
where $L_m^2 = \max\{ 2C_1^2 + 6N M^2 C_2^2, 2C_1^2(\frac{1}{N}+n^2C_3^2) + 6 M^2 C^2_2(n^2+1+n^2C_3^2N) + 2n^2B^2C_3^2N\}$. The second inequality follows from the triangle inequality, the uniform boundedness and the Lipschitz condition of $\beta_{i}$ and $\phi^{\ell}$, 
\[
\begin{aligned}
 &|v^{\ell} \phi^{\ell}(v) \beta(x, L_N(v)) - u^{\ell} \phi^{\ell}(u) \beta(y, L_N(u))|^2\\
=& \, |\phi^{\ell}(v) [v^{\ell} \beta(x, L_N(v)) - u^{\ell} \beta(y, L_N(u))] + [\phi^{\ell}(v) -  \phi^{\ell}(u)] u^{\ell} \beta(y, L_N(u))|^2\\
 \leq & \, 2 |v^{\ell} \beta(x, L_N(v)) - u^{\ell} \beta(y, L_N(u))|^2 + 2B^2 |\phi^{\ell}(v) -  \phi^{\ell}(u)|^2  \\
 \leq & \, 6 C_2^2 [|x-y|^2 + n^2|v^{\ell}-u^{\ell}|^2+ \mathcal{W}_2^2(L_N(v),L_N(u)] + 2 n^2 B^2 C_3^2 |v-u|^2. 
\end{aligned}
\]


Thus, we get the Lipschitz continuity of $F(\cdot)$. In the same vein, we conclude the Lipschitz continuity of $G(\cdot)$. Thus, according to the existence and uniqueness conditions of McKean-Vlasov dynamics in \cite{lecturemfg}, the system \eqref{nparticle}-\eqref{nv} is well-defined.
\end{proof}

\section{Derivations related to fixed point problem}
\label{apdx:fp}

We use this section to clarify the fixed-point mapping in Figure~\ref{fig: fixedpt}. {This also specifies the fixed point of coefficients $(\gamma, \tau)$, which is needed to check their Lipschitz conditions in assumption~\ref{asmp1:main}.}

To generalize Assumption~\ref{asmp1:main}, we instead assume that the coefficients $\gamma(\cdot)$ and $\tau(\cdot)$ of the trading volume processes $\mathcal{Y}_i(t)$ are local Lipschitz continuous with respect to $(\mathbf{x}, \mathbf{y})$ uniformly with respect to time $t$ on every compact interval, i.e., for any $T, M > 0$, there exists $k > 0$ such that for any $\mathbf{x},  \mathbf{x}', \mathbf{y}, \mathbf{y}'  \in \mathbb R^n, \quad $
\begin{equation}
\label{eq:locallip_tau}
    \sup_{t \in [0,T]} \left( |\gamma(t, \mathbf{x}, \mathbf{y}) - \gamma(t, \mathbf{x}', \mathbf{y}')| + |\tau(t, \mathbf{x}, \mathbf{y}) - \tau(t, \mathbf{x}', \mathbf{y}')| \right) \leq k \left(|\mathbf{x}-\mathbf{x}'| + |\mathbf{y}-\mathbf{y}'|\right), \quad |\mathbf{x}|,  |\mathbf{x}'|, |\mathbf{y}|, |\mathbf{y}'|<M.
\end{equation}
We can ensure the well-posedness of the market system of $(\mathcal{X}(\cdot), \mathcal{Y}(\cdot))$ by \cite[Theorem 12.3, Lemma 12.4]{russo2022stochastic}. The regularity assumption of $u^{\ell}(\cdot)$ in Theorem 4.1 indicates that $\phi(\cdot) \in C^{1,2,2}([0,T] \times (0,\infty)^n \times (0,\infty)^n)$, where {$\phi^{\ell}(t, \mathbf{x}, \mathbf{y}): [0,T] \times \mathbb{R}_+^n \times \mathbb{R}_+^n \rightarrow \mathbb{R}^n$}, $\ell = 1, \ldots, N$. In addition, $u(\cdot)$ is generalized to the solution of a Cauchy problem with time-inhomogeneous coefficients.
\[
d \mathcal Y_i(t) = \frac{1}{N} \sum_{\ell=1}^N d \big[ V^{\ell}(t) \phi^{\ell}_i(\mathcal{X}(t), \mathcal{Y}(t)) \big] = \gamma_i(t, \mathcal{X}(t), \mathcal{Y}(t)) dt + \sum_{k=1}^n \tau_{ik}(t, \mathcal{X}(t), \mathcal{Y}(t)) dW_k(t),
\]
where the coefficients $\gamma(\cdot)$ and $\tau(\cdot)$ can be determined by It\^o's formula on $\phi^{\ell}(\mathbf{x}, \mathbf{y}): \mathbb{R}_+^n \times \mathbb{R}_+^n \rightarrow \mathbb{R}^n$, i.e.,
\begin{equation}
\label{gonc}
\begin{aligned}
 & d \big[V^{\ell}(t) \phi^{\ell}_i(t, \mathbf{x}, \mathbf{y}) \big] \\
= & V^{\ell}(t) \Bigg( \phi_i^{\ell}(t, \mathbf{x}, \mathbf{y}) \phi^{\ell}(t, \mathbf{x}, \mathbf{y}) \beta(\mathbf{x}, \mathbf{y}) + D_t \phi_i^{\ell}(t, \mathbf{x}, \mathbf{y}) + \frac{1}{2} \sum_{p,q = 1}^{2n} \widehat{a}_{pq}(\mathbf{x}, \mathbf{y}) \partial^2_{pq} \phi_i^{\ell}(t, \mathbf{x}, \mathbf{y})\\
&  \hspace{1cm} + \sum_{p=1}^{n} \beta_p(\mathbf{x}, \mathbf{y}) D_{p} \phi_i^{\ell}(t, \mathbf{x}, \mathbf{y}) + \sum_{p=n+1}^{2n} \gamma_{p-n}(t, \mathbf{x}, \mathbf{y}) D_{p} \phi_i^{\ell}(t, \mathbf{x}, \mathbf{y}) \Bigg) dt\\
& + V^{\ell}(t) \left( \phi_i^{\ell}(t, \mathbf{x}, \mathbf{y}) \phi^{\ell}(t, \mathbf{x}, \mathbf{y}) \sigma(t) + \sum_{p=1}^{2n} \widehat{\sigma}_p(\mathbf{x}, \mathbf{y}) D_{p} \phi_i^{\ell}(t, \mathbf{x}, \mathbf{y}) \right) dW_t \Bigg |_{(\mathbf{x}, \mathbf{y}) := (\mathcal{X}(t), \mathcal{Y}(t))},
\end{aligned}    
\end{equation}
or for the time homogeneous case,
\begin{equation}
\label{gonc1}
\begin{aligned}
 & d \big[V^{\ell}(t) \phi^{\ell}_i(\mathbf{x}, \mathbf{y}) \big] \\
= & V^{\ell}(t) \Bigg( \phi_i^{\ell}(\mathbf{x}, \mathbf{y}) \phi^{\ell}(\mathbf{x}, \mathbf{y}) \beta(\mathbf{x}, \mathbf{y}) + \frac{1}{2} \sum_{p,q = 1}^{2n} \widehat{a}_{pq}(\mathbf{x}, \mathbf{y}) \partial^2_{pq} \phi_i^{\ell}(\mathbf{x}, \mathbf{y})\\
&  \hspace{1cm} + \sum_{p=1}^{n} \beta_p(\mathbf{x}, \mathbf{y}) D_{p} \phi_i^{\ell}(\mathbf{x}, \mathbf{y}) + \sum_{p=n+1}^{2n} \gamma_{p-n}(\mathbf{x}, \mathbf{y}) D_{p} \phi_i^{\ell}(\mathbf{x}, \mathbf{y}) \Bigg) dt\\
& + V^{\ell}(t) \left( \phi_i^{\ell}(\mathbf{x}, \mathbf{y}) \phi^{\ell}(\mathbf{x}, \mathbf{y}) \sigma(t) + \sum_{p=1}^{2n} \widehat{\sigma}_p(\mathbf{x}, \mathbf{y}) D_{p} \phi_i^{\ell}(\mathbf{x}, \mathbf{y}) \right) dW_t \Bigg |_{(\mathbf{x}, \mathbf{y}) := (\mathcal{X}(t), \mathcal{Y}(t))},
\end{aligned}    
\end{equation}
where $\widehat{\sigma}(\cdot)$ and $\widehat{a}(\cdot)$ are defined in \eqref{eq: sigmacombine}. $D_{p} \phi_i^{\ell}(t, \mathbf{x}, \mathbf{y}) := \frac{\partial \phi_i^{\ell}(t, \mathbf{x}, \mathbf{y})}{\partial x_p}$, for $p =1, \ldots, n$; $D_{p} \phi_i^{\ell}(t, \mathbf{x}, \mathbf{y}) := \frac{\partial \phi_i^{\ell}(t, \mathbf{x}, \mathbf{y})}{\partial y_{p-n}}$, for $p =n+1, \ldots, 2n$. Thus, equating the drift (diffusion, respectively) terms to $\gamma(\cdot)$ (to $\tau(\cdot)$, respectively), we get the fixed point condition of $\phi(t, \mathbf{x}, \mathbf{y})$, and we can specify the drift and diffusion coefficients $\tau(\cdot)$ and $\gamma(\cdot)$ of the fixed point solution $\phi(\cdot)$.

\end{appendices}

\end{document}